\documentclass[aps, pra,  onecolumn, nofootinbib, footnote, superscriptaddress, longbibliography]{revtex4-2}

\usepackage{setspace} % Required package
\usepackage{soul}
\usepackage{relsize}
\usepackage{lmodern}
\usepackage{slantsc}
\usepackage{bbm}
\usepackage{graphicx}
\usepackage{amsmath}
\usepackage{amssymb}
\usepackage{amsthm}
\usepackage{amsbsy}
\usepackage{mathrsfs}
\usepackage{mathtools}
\usepackage{varioref}
\usepackage{dsfont}
\usepackage{bm}
\usepackage{centernot}
\usepackage[normalem]{ulem}
\usepackage[dvipsnames,dvipsnames]{xcolor}
\definecolor{myblue}{rgb}{0.2,0.2,0.8}
\definecolor{myblack}{rgb}{0,0,0}
\definecolor{myurl}{rgb}{0.1,0.1,0.4}
\usepackage[colorlinks=true,citecolor=myblue,linkcolor=orange,urlcolor=myurl]{hyperref}
\usepackage{cleveref}
\usepackage[font = footnotesize, labelfont = bf, justification = RaggedRight]{caption}
\usepackage{subfigure}
\usepackage{booktabs} % for much better looking tables
\usepackage{array} % for better arrays (eg matrices) in maths
\usepackage{paralist} % very flexible & customisable lists (eg. enumerate/itemize, etc.)
\usepackage{verbatim} % adds environment for commenting out blocks of text & for better verbatim
\edef\restoreparindent{\parindent=\the\parindent\relax}
\usepackage{multirow}

\usepackage[parfill]{parskip}
\setul{}{20 pt}

\newtheorem{lemma}{Lemma}
\numberwithin{lemma}{section}
\newcommand{\lemref}[1]{Lemma~\ref{#1}}
\newtheorem{prop}{Proposition}
\numberwithin{prop}{section}
\newcommand{\propref}[1]{Proposition~\ref{#1}}
\newtheorem{corollary}{Corollary}
\numberwithin{corollary}{section}
\newcommand{\corref}[1]{Corollary~\ref{#1}}

\numberwithin{example}{section}

\newtheorem{theorem}{Theorem}
\numberwithin{theorem}{section}
\newcommand{\thmref}[1]{Theorem~\ref{#1}}

\theoremstyle{definition}
\newtheorem{definition}{Definition}
\newcommand{\defref}[1]{Definition~\ref{#1}}
\newtheorem*{remark}{Remark}
\usepackage{pifont}

\DeclareMathOperator\supp{supp}

\newcommand{\<}{\langle}
\renewcommand{\>}{\rangle}

\newcommand{\rr}{{\mathcal{R}}}
\newcommand{\lo}{{\mathcal{L}}}

\newcommand{\s}{{\mathcal{S}}}
\newcommand{\ee}{{\mathcal{E}}}

\newcommand{\xx}{{\mathcal{X}}}

\renewcommand{\aa}{{\mathcal{A}}}

\newcommand{\ff}{{\mathcal{F}}}

\newcommand{\opset}{{\mathscr{O}}}
\newcommand{\inset}{{\mathscr{I}}}
\newcommand{\ii}{{\mathcal{I}}}
\newcommand{\I}{{\mathsf{I}}}
\newcommand{\II}{{\mathsf{II}}}
\newcommand{\III}{{\mathsf{III}}}
\newcommand{\jj}{{\mathcal{J}}}
\newcommand{\co}{\mathds{C}}
\newcommand{\re}{\mathds{R}}

\newcommand{\h}{{\mathcal{H}}}
\newcommand{\kk}{{\mathcal{K}}}
\newcommand{\hs}{{\mathcal{H}\sub{\s}}}
\newcommand{\ha}{{\mathcal{H}\sub{\aa}}}

\newcommand{\hsys}{{H\sub{\s}}}
\newcommand{\happ}{{H\sub{\aa}}}

\newcommand{\E}{\mathsf{E}}

\newcommand{\Z}{\mathsf{Z}}

\newcommand{\A}{\mathsf{A}}

\renewcommand{\P}{\mathsf{P}}

\newcommand{\cT}{\mathbb{T}}
\newcommand{\cS}{\mathbb{S}}

\renewcommand{\nat}{\mathds{N}}

\newcommand{\one}{\mathds{1}}
\newcommand{\onesys}{\mathds{1}\sub{\s}}
\newcommand{\oneapp}{\mathds{1}\sub{\aa}}
\newcommand{\idch}{\operatorname{id}}
\newcommand{\idchsys}{\operatorname{id}\sub{\s}}
\newcommand{\idchapp}{\operatorname{id}\sub{\aa}}
\newcommand{\zero}{\mathds{O}}

\newcommand{\tr}{\mathrm{tr}}
\newcommand{\T}{\mathrm{T}}
\newcommand{\tra}{\mathrm{tr}\sub{\aa}}
\newcommand{\trs}{\mathrm{tr}\sub{\s}}

\newcommand{\av}{\mathrm{av}}

\makeatletter
\newcommand*\bigcdot{ \mathpalette\bigcdot@{.5} }
\newcommand*\bigcdot@[2]{\mathbin{\vcenter{\hbox{\scalebox{#2}{$\m@th#1\bullet$}}}}}
\makeatother

\newcommand{\sub}[1]{_{\!\mathsmaller{\, #1}}}

\newcommand{\eq}[1]{Eq.~\eqref{#1}}

\newcommand{\sect}[1]{Sec.~\ref{#1}}
\newcommand{\app}[1]{Appendix~(\ref{#1})}
\newcommand{\ket}[1]{|{#1}\rangle}
\newcommand{\bra}[1]{\langle{#1}|}

\newcommand{\rank}[1]{\mathrm{rank}\left( {#1}\right)}
\newcommand{\proj}[1]{\ensuremath{\left|#1\right\rangle\!\!\left\langle#1\right|}}

\definecolor{darkgreen}{RGB}{0,175,0}
\definecolor{pblue}{RGB}{46,117,182}
\definecolor{pred}{RGB}{197, 90, 17}

\begin{document}

\title{A hierarchy of thermodynamically consistent quantum operations}

\author{Fereshte Shahbeigi  }
\email{fereshte.shahbeigi@savba.sk}
\affiliation{RCQI, Institute of Physics, Slovak Academy of Sciences, D\'ubravsk\'a cesta 9, Bratislava 84511, Slovakia}

\author{M. Hamed Mohammady  }
\email{m.hamed.mohammady@savba.sk}
\affiliation{RCQI, Institute of Physics, Slovak Academy of Sciences, D\'ubravsk\'a cesta 9, Bratislava 84511, Slovakia}

\begin{abstract}
In order to determine what quantum operations and measurements are consistent with the laws of thermodynamics, one must start by allowing all  processes allowed by the framework of quantum theory, and then impose the laws of thermodynamics as a  set of constraints. Here, we consider a hierarchy of quantum operations and measurements that are consistent with ($\I$) the weak third law, ($\II$) the strong third law, and ($\III$) both the second and the third laws of thermodynamics, i.e., operations and measurements that are fully consistent with thermodynamics. Such characterisation allows us to identify which particular thermodynamic principle is responsible for the (un)attainability of a given quantum operation or measurement. In the case of channels, i.e.,  trace-preserving operations, we show that a channel belongs to ($\I$) and ($\II$) if and only if it is strictly positive and rank non-decreasing, respectively, whereas a channel belongs to ($\III$) only if it is rank non-decreasing and does not perturb a strictly positive state. On the other hand, while thermodynamics does not preclude the measurability of any POVM, the realisable state-update rules for measurements are increasingly restricted as we go from ($\I$) to ($\III$).

\end{abstract}

\maketitle

%-------------------------------------------------------------
%-----Section: Introduction-----------------------------------
%-------------------------------------------------------------
\section{Introduction}
\label{sec:intro}

The most general way in which a quantum system may transform, potentially probabilistically, is described by a completely positive (CP) trace non-increasing map, known as  an operation. An instrument is a family of operations that sum to a deterministic (trace-preserving)  operation,  called a channel, and provides the most general state-update rule for a quantum measurement. The mathematical formalism of quantum theory allows for every possible   operation, or measurement,  on a quantum system of interest to be \emph{purified}. In other words, every operation or measurement may be  dilated into a unitary channel acting on the compound of the system of interest and an auxiliary quantum system  initially prepared in a pure state, followed by readout of a pointer observable on the auxiliary system     \cite{Von-Neumann-Foundations, Neu40, Stinespring1955, Ozawa1984}. However, for such pure dilations to be interpreted as physical  processes and not just formal mathematical constructs, that is, as physical interactions between the system and an existent environment or measuring apparatus, they must be consistent with physical principles beyond quantum theory alone: in particular, they must be consistent with the laws of thermodynamics. While  pure dilations are consistent with the second law of thermodynamics \cite{Purves-2020, Minagawa2023a}, they are in conflict with the third law which states that it is impossible to cool a system to absolute zero temperature, and hence  prohibits  the preparation of  quantum systems in pure states; in fact, the third law permits quantum systems to be prepared only in strictly positive states, i.e., states with full rank \cite{Schulman2005, Allahverdyan2011a, Ticozzi2014,Masanes2014,Wilming2017a,Scharlau2016a,  Freitas2018a, Taranto2021, Taranto2025a}.

An operation or measurement is thermodynamically consistent, therefore, if it admits a thermodynamically consistent process. That is to say, the operation or measurement must be realisable by a thermodynamically permissible apparatus state preparation, and a permissible interaction between system and apparatus. In the first analysis, a process is thermodynamically consistent precisely when it utilises a unitary interaction with an apparatus prepared in a strictly positive state. Such thermodynamically consistent operations can be seen as generalisations of so-called thermal operations, implemented by an energy conserving unitary interaction with an apparatus prepared in thermal equilibrium, which is a strictly positive state \cite{Horodecki2013, Navascues2015a, Brandao2015a, Perry2015, Lostaglio2016, Mazurek2018, Mohammady2022}. It follows that while thermal operations can be interpreted as those that do not consume any thermodynamic resources, thermodynamically consistent operations can be interpreted as those that consume only finite resources. It has been shown that some textbook examples of operations and measurements do not admit such a notion of a thermodynamically consistent process  \cite{Terhal-mixed-environment, Zalka-mixed-environment, Karol-unistochastic, GOUR2015, Guryanova2018,  Panda2023, Vetrivelan2023, Debarba2024}.  

While the restriction of the apparatus state preparation to a strictly positive one has an unambiguous operational justification in terms of quantum theory and thermodynamics \emph{alone}, the same cannot be said for  unitarity of the interaction. To be sure, if the compound of system and apparatus is assumed to be  \emph{informationally} closed so that the dynamics is reversible, then the interaction channel must be unitary, in accordance with the conventional wisdom. However, if we assume only that the   compound is \emph{thermodynamically closed} so that no heat is exchanged with an external environment---a precondition for the process to be amenable to thermodynamic analysis \cite{Mohammady2025}---then the interaction  is consistent with the second law if and only if it is described by a bistochastic channel, i.e., a unital and trace preserving CP map. This is because bistochastic channels are precisely those that do not decrease the entropy of any state and so cannot be used to construct a \emph{perpetuum mobile} \cite{Purves-2020}.  In other words, since unitary channels are a subclass of bistochastic ones, then unitarity of the interaction (together with a strictly positive apparatus preparation) is sufficient, but not necessary, for thermodynamic consistency of a process realising a given operation or measurement: Unitarity of the interaction is logically necessitated only if physical principles beyond the mathematical framework of operational quantum theory and thermodynamics are taken into account.   Furthermore, if we relax the requirement that the interaction must be consistent with the second law, but still demand that it be consistent with the third, we may generalise the permissible interactions  beyond the class of bistochastic channels. Note that state preparations are in fact channels that take a quantum system, initially prepared in an arbitrary state,  to a fixed known  state. As such,   the statement of the third law---that a quantum system can be prepared only in a strictly positive state---can be expressed in terms of the properties of the channels that may be physically implemented. Here, we are left with two possible definitions for channels that are consistent with the third law, with one a stronger form of the other. The minimal requirement for a channel to be consistent with the third law is that such channel must be strictly positive, i.e., it must map every strictly positive state to a strictly positive state  \cite{VomEnde2022a, Mohammady2022a}. Such channels cannot prepare a system, initially given in some state with full rank (such as a thermal state) in a pure state.  We call strictly positive channels as those that are consistent with the weak third law. On the other hand, a stronger condition for consistency with the third law is that the channel must also be  rank non-decreasing, i.e., it must not decrease the rank of any state. Such channels can prepare a system in a pure state only if the system is initially in a pure state. Rank non-decreasing channels are a proper subset of strictly positive channels, and we call rank non-decreasing channels as those that are consistent with the strong third law. Note that since bistochastic channels are also rank non-decreasing, then they are in fact consistent with both the second and the (strong) third laws.    

In order to gain a better understanding of what particular thermodynamic law is responsible   for the (un)attainability of a given operation or measurement, in this paper we consider the following hierarchy of thermodynamically consistent processes: Processes that are ($\I$) consistent with the weak third law; ($\II$)  consistent with the strong third law; and ($\III$)  consistent with both the second and the third laws, and hence fully consistent with thermodynamics. In each class of the hierarchy, the apparatus state preparation is strictly positive, whereas the interaction channels in ($\I$), ($\II$), and ($\III$) are strictly positive, rank non-decreasing, and bistochastic, respectively.  We do not consider the second law in isolation because, as stated above, all operations and measurements admit a process that is consistent with the second law alone \cite{Minagawa2023a}. We then characterise the corresponding hierarchy of thermodynamically consistent operations and measurements, providing  necessary (and in some cases also sufficient) conditions for an operation or measurement to belong to each class in the hierarchy.  For example, a channel is consistent with the weak  (strong) third law if and only if it is strictly positive (rank non-decreasing). On the other hand, a channel is fully consistent with thermodynamics only if it is both rank non-decreasing \emph{and}   has a strictly positive fixed state. Indeed, this provides further evidence that thermodynamically consistent operations are true generalisations of thermal operations: thermal channels are known to preserve the thermal state of the system, and if we relax the energy conservation of the interaction and thermality of the initial environment preparation, then while the thermal equilibrium state of the system may be perturbed, the system continues to have a strictly positive non-equilibrium  steady state. 

The paper is structured as follows. In \sect{sec:preliminaries}  we establish notation and review some basic facts about operational quantum physics and the theory of quantum measurements. Readers familiar with these topics may skip directly to \sect{sec:consistent-processes}, where we define the three hierarchies of thermodynamically consistent processes.  \sect{sec:results} contains the main results of our paper, with  \sect{sec:results-operations} characterising the set of operations that may be realised by processes in each class of the hierarchy, while  \sect{sec:results-instruments} concerns instruments  and in particular their non-disturbance properties.     We conclude with some discussion in \sect{sec:discussion}.

%-------------------------------------------------------------
%-----Section: Preliminaries----------------------------------
%-------------------------------------------------------------

\section{preliminaries}\label{sec:preliminaries}
In this section, we shall cover the basics of operational quantum physics and the theory of quantum measurement. For more details, see, e.g., Refs. \cite{PaulBuschMarianGrabowski1995, Busch1996, Heinosaari2011, Wolf2012, Busch2016a,  Hayashi-QIT}. Readers familiar with these topics may skip this section and proceed to \sect{sec:consistent-processes}.
\subsection{Basic concepts}

We always consider systems with a complex Hilbert space $\h$ of finite dimension. Let $\lo(\h)$ be the algebra of linear operators on $\h$, with $\one$ and $\zero$ denoting the unit and null operators in $\lo(\h)$, respectively.  An operator $ A \in  \lo(\h)$ is called positive definite, or strictly positive, if $A > \zero$, i.e., if  all the eigenvalues of $A$ are strictly positive, which implies that $\rank{A} = \dim(\h)$.  For any $A > \zero$ and $B \in \lo(\h)$, it holds that $\tr[A \, B^*B] = \zero \iff B = \zero$. An operator $E \in \lo(\h)$ such that $\zero \leqslant E \leqslant \one$ is called an effect. An effect is trivial if it is proportional to the identity, and is non-trivial otherwise.   $E$ is   called a  norm-1 effect  if $\|E \| =1$,  where $\| \bigcdot \|$ is the operator norm;  a norm-1 effect has at least one eigenvector with eigenvalue one. A subclass of norm-1 effects are projections, which satisfy  $E^2 = E$. An effect $E < \one$ does not have  the norm-1 property.   An effect $E$ is indefinite (or completely unsharp) if it is strictly positive and lacks the norm-1 property, i.e., if $\zero < E < \one$, which means that the spectrum of  $E$ does not contain  zero or one. See \app{app:definite-effects} for further details.   A state on $\h$ is defined as a positive semidefinite operator $\rho$ of unit trace, with $\s(\h)$ denoting the state space on $\h$.  For any subset $\mathscr{A} \subseteq \lo(\h)$, we define the \emph{commutant} of $\mathscr{A}$ in $\lo(\h)$ as
\begin{align*}
    \mathscr{A}' \coloneq \{B \in \lo(\h) : [B, A] = \zero \quad \forall \, A \in \mathscr{A}\} \, .
\end{align*}

\subsection{Operations and channels}

In the ``Schr\"odinger picture'', a completely positive (CP)  linear map $\Phi : \lo(\h) \to \lo(\kk)$ is called an operation if it is trace non-increasing. When $\kk = \h$, we say that the operation acts in $\h$, and we denote $\opset(\h)$ as the set of operations acting in $\h$.       Among the operations are channels, which preserve the trace. The identity channel acting in $\h$ is denoted by $\idch$, which maps all operators to themselves.   For each CP map $\Phi : \lo(\h) \to \lo(\kk)$ there exists a unique ``Heisenberg picture'' dual $\Phi^* : \lo(\kk) \to \lo(\h)$ defined by the trace duality  $\tr[\Phi^*(A) B] = \tr[A \Phi(B)]$ for all $A \in \lo(\kk),B \in \lo(\h)$. $\Phi^*$ is also a CP map, and if $\Phi$ is an operation, then $\Phi^*$ is subunital, i.e., $\Phi^*(\one\sub{\kk}) = E \leqslant \one\sub{\h}$ is an effect; we say that the operation $\Phi$ is compatible with $E$. The dual of a channel is  unital, i.e., channels are compatible with the unit effect $\one\sub{\h}$.  We denote the parallel application of two CP maps $\Phi_i : \lo(\h_i) \to \lo(\kk_i)$, $i = 1,2$, as $\Phi_1 \otimes \Phi_2 : \lo(\h_1 \otimes \h_2) \to \lo(\kk_1 \otimes \kk_2), A_1 \otimes A_2 \mapsto \Phi_1(A_1) \otimes \Phi_2(A_2)$, and if $\kk_1 = \h_2$, the sequential application as $\Phi_2 \circ \Phi_1 : \lo(\h_1) \to \lo(\kk_2), A \mapsto \Phi_2[\Phi_1(A)]$.

A CP map $\Phi : \lo(\h) \to \lo(\kk)$ is called strictly positive if $ A >\zero \implies \Phi(A) > \zero$. In the case where $\kk = \h$, $\Phi$  is rank non-decreasing if $\rank{\Phi(A)} \geqslant \rank{A}$ for all $A \geqslant \zero$. $\Phi$ is rank non-decreasing if and only if $\Phi^*$ is rank non-decreasing (\lemref{lemma:equivalent-dual-positive-rank-non-dec}).  While all rank non-decreasing  maps are strictly positive, not all strictly positive  maps are rank non-decreasing, see an example in \cite[Appendix B]{Mohammady2022a}.    A channel $\Phi$ acting in $\h$ is called bistochastic if it preserves both the trace and the unit; $\Phi$ is bistochastic if and only if $\Phi^*$ is bistochastic.  While all bistochastic channels are rank non-decreasing,  not all rank non-decreasing channels are bistochastic. See \app{app:completely-rank-non-decreasing} for a detailed discussion on strictly positive and rank non-decreasing maps.

\subsection{Fixed points of operations}

We define the fixed point sets  of an operation $\Phi$ acting in a system $\h$, and its dual $\Phi^*$, as
\begin{align*} %\label{eq:fixed-points}
   &\ff(\Phi) \coloneq \{A \in \lo(\h) : \ \Phi(A) = A\} \,  ,  &\ff(\Phi^*) \coloneq \{A \in \lo(\h) :\ \Phi^*(A) = A\}\, .
\end{align*}
$\ff(\Phi)$ and $\ff(\Phi^*)$ are closed under linear combination and involution, 
 and they have the same dimension, where the dimension of a subset $\mathscr{A} \subseteq \lo(\h)$ is equal to  the smallest number of linearly independent operators that spans $\mathscr{A}$.   If $\Phi$ is a channel then by the Schauder–Tychonoff fixed point theorem \cite{Fixed-points-appl, Tumulka2024} $\ff(\Phi)$ contains at least one state $\rho_0$. On the other hand, an $E$-compatible  operation $\Phi$ acting in $\h$ has non-vanishing fixed points if and only if there exists a projection $P $ such that $EP = P$, which implies that $\| E\| =1$ must hold,  and  $ P \Phi(P \bigcdot P) P  = \Phi(P \bigcdot P)$.   See \app{app:definite-effects} for further details.

\subsection{Observables,  instruments, and measurement processes}

An observable on  $\h$  is represented by a normalised positive operator valued measure (POVM) $\E$. We consider only discrete observables,  which may be represented as a family of effects $\E \coloneq \{E_x : x \in \xx\}$   on $\h$ such that $\sum_{x\in \xx} E_x = \one$. Here,  $\xx \coloneq \{x_1, \dots, x_N \}$ is a discrete (and finite) value space (or space of  measurement outcomes). The probability of observing outcome $x$ when measuring $\E$ in the state $\rho$ is given by the Born rule as $p^E_\rho(x) \coloneq \tr[\rho \, E_x]$.  Without loss of generality, we shall consider only observables such that $E_x \ne \zero$ for any $x$. Since an outcome $x$ for which $E_x = \zero$ is  observed with probability zero, this can always be done by replacing the original value space $\xx$ with the relative complement $\xx \backslash \{x : E_x = \zero\}$. We define the following classes of observables:

\

\begin{definition}[Observables]\label{defn:observable-class}
    Let $\E \coloneq \{E_x : x\in \xx\}$ be an observable. 
\begin{enumerate}[(i)]
    \item $\E$ is a commutative observable if $\E \subset \E'$, i.e., if  $[E_x,E_y]=\zero$ for all $x,y \in \xx$.

    \item  $\E$ is a sharp  (or projection valued) observable if $E_x E_y = \delta_{x,y} E_x$ for all $x,y \in \xx$, i.e., if all the effects  are mutually orthogonal projections. An observable that is not sharp is called unsharp.

    \item $\E$ is a norm-1 observable if $\|E_x\| = 1$ for all $x \in \xx$.

    \item $\E$ is an indefinite (or completely unsharp) observable   if $\zero < E_x < \one$ for all $x \in \xx$. 
\end{enumerate}
 \unskip\nobreak\hfill $\square$   
\end{definition}

Note that  sharp observables are both commutative and norm-1. On the other hand, while a norm-1 (indefinite) observable may also be commutative, it cannot be indefinite (norm-1). Moreover, while a norm-1 observable admits for every outcome $x$  a state $\rho$ such that $p^\E_\rho(x) \in \{0,1\}$, for an indefinite observable it holds that $0 < p^\E_\rho(x) < 1$ for all  $\rho$ and $x$. That is, norm-1 observables admit states for which the outcome of measurement can be predicted with probabilistic certainty, whereas this possibility does not exist for an indefinite observable.  

A discrete instrument (or normalised operation valued measure) acting in $\h$, with outcomes in $\xx$,  is represented by a family of operations  $\ii = \{\ii_x \in \opset(\h)  : x\in \xx\}$  such that $\ii_\xx(\bigcdot) \coloneq \sum_{x \in \xx} \ii_x(\bigcdot)$ is a channel. We denote the set of instruments acting in $\h$ as $\inset(\h)$. Every instrument is compatible with a unique observable $\E$ via $\ii_x^*(\one) = E_x$;    we shall refer to such an instrument $\ii$, and to the corresponding channel $\ii_\xx$, as $\E$-compatible. As above, we  consider only instruments such that $\ii_x^*(\onesys) = E_x \ne \zero$ for any $x$. Note that if $|\xx| = N = 1$, then the instrument has just one operation $\ii_x = \ii_\xx$, which is in fact a channel, and is thus compatible with a trivial observable $\E = \{E_x = \onesys\}$.

Let $\hs$ be a  system of interest. Let $\ha$ be an auxiliary system (an apparatus or environment) with $\xi$  a state on $\ha$;  let   $\ee$ be  an interaction channel acting in $\hs \otimes \ha$; and let $\Z\coloneq \{Z_x: x \in \xx\}$ be an observable on $\ha$ with the same value space $\xx$ as that of the observable to be measured in $\hs$. The tuple $(\ha, \xi, \ee, \Z)$ is a \emph{measurement process} for the instrument $\ii$ if  all operations of $\ii$ may be written as
\begin{align}\label{eq:instrument-implementation}
    \ii_x(\bigcdot) = \tra[\onesys \otimes Z_x \  \ee(\bigcdot \otimes \xi)] \qquad  \forall \,  x \in \xx \, .
\end{align}
The physical interpretation of the above is as follows: we first couple the system of interest with the apparatus, prepared in some fixed state $\xi$,  and let them interact via the channel $\ee$. Subsequently, we \emph{post-select} the apparatus by a measurement of $\Z$, so that conditional on observing an outcome $x$  associated with the effect $Z_x$, the system will transform via $\ii_x$.  By the Naimark-Ozawa dilation theorem \cite{Ozawa1984} every instrument in $\inset(\hs)$ (and hence every operation in $\opset(\hs)$)  admits \emph{some} process: choose $\xi$ to be pure, $\ee$ to be unitary, and $\Z$ to be  projection-valued. Since $\hs$ is assumed to be finite-dimensional, then $\ha$ can always be chosen to be  finite.   However, in general $\xi$ need not be pure, $\ee$ need not be unitary, and $\Z$ need not be projection-valued.

%-------------------------------------------------------------
%-----Section: consistent-processes---------------------------
%-------------------------------------------------------------

\section{Thermodynamically consistent processes}\label{sec:consistent-processes}

In this paper, we wish to determine the properties of the   operations and instruments one may realise by a measurement process $(\ha, \xi, \ee, \Z)$ as in \eq{eq:instrument-implementation},  with the only constraints  being that the process implementing them must be consistent with  thermodynamic principles: in particular, the second and  third laws of thermodynamics. To be sure, a process is fully consistent with the laws of thermodynamics if it is consistent with the conjunction of all thermodynamical laws. But in order to delineate what particular law is responsible for the (un)attainability of a given operation or instrument, we establish the following  \emph{hierarchy} of thermodynamically consistent processes:

\

\begin{definition}[Thermodynamically consistent processes]
\label{defn:thermo-consistent-process}

Let $(\ha, \xi, \ee, \Z)$ be a process. We say that the process is:

\begin{enumerate}[(I)]

    \item   consistent with the weak third law if  $\xi$ is a strictly positive state and $\ee$ is a  strictly positive channel. 

    \item  consistent with the strong third law if   $\xi$ is a strictly positive state and $\ee$ is a  rank non-decreasing channel. 

     \item fully consistent with thermodynamics if   $\xi$ is a strictly positive state and $\ee$ is a  bistochastic channel. 
\end{enumerate}
\unskip\nobreak\hfill $\square$
\end{definition}

\

{\bf The weak third law: }The third law of thermodynamics, or Nernst's unattainability principle, states that it is impossible to cool a system to absolute zero temperature with finite resources of time, energy, or control complexity \cite{Schulman2005, Allahverdyan2011a, Freitas2018a, Taranto2021}. In finite dimensions, a system $\hs \otimes \ha$ with Hamiltonian $H$ at thermal equilibrium with respect to temperature $T$ is described by a Gibbs state $\tau(T) \coloneq e^{- H/ k_B T} / \tr[e^{- H/ k_B T}]$.  $\tau(T)$ is strictly positive whenever $T > 0$, and (provided a non-trivial Hamiltonian) is rank deficient when $T=0$. Consequently, a channel $\ee$ acting in $\hs \otimes \ha$ is  consistent with the third law  only if $\ee (\tau(T)) > \zero$ whenever $T >0$. Since the existence of a strictly positive operator in the image of a positive linear map is equivalent to the strict positivity of such a map \cite{VomEnde2022a}, a minimal requirement for a channel (with potentially different input and output spaces) to be consistent with the third law---the weak third law---is that such channel must be strictly positive \cite{Mohammady2022a}.   Indeed, such a characterisation already ensures that the only state preparations  $\xi$ on $\ha$ that are consistent with the  third law are strictly positive: A state preparation $\xi$ on $\ha$ is characterised as a channel which sends a trivial system $\co^1 \equiv \co |\Omega\>$  to $\ha$, i.e.,  $\Xi : \lo(\co^1) \to \lo(\ha), \proj\Omega \mapsto \xi$. Since $\proj{\Omega}$ is strictly positive in $\lo(\co^1)$, then strict positivity of the channel $\Xi$ ensures that the state $\xi$ is strictly positive.

{\bf The strong third law: } Note that provided a rank-deficient input state $\rho$ on $\hs \otimes \ha$, it is possible to have $\rank{\ee(\rho)} < \rank{\rho}$ even if $\ee$ is a strictly positive channel acting in $\hs \otimes \ha$ \cite[Appendix B]{Mohammady2022a}. As such, a stronger form of the third law---applicable now only to the case where a channel's input and output systems are the same---would be to demand that $\ee$ must be rank non-decreasing. The distinction between the weak and the strong third law can be given  the following operational interpretation: a channel   that is consistent with the weak third law cannot prepare a pure output from a strictly positive input. On the other hand, a channel  that is consistent with the strong third law can prepare a pure output only from a pure input. Note that for such a characterisation to be physically meaningful as a \emph{law of nature}, then  a local application of a rank non-decreasing channel should not reduce the rank of a global entangled state; all extensions $\ee \otimes \idch$, with $\idch$ the identity channel acting in an arbitrary finite system, must also be rank non-decreasing, i.e., $\ee$ must be \emph{completely rank non-decreasing}. The reasoning is analogous to why a positive map is physical only if it is completely positive. While it is well-known that $\ee \otimes \idch$ is strictly positive (or bistochastic) if $\ee$ is strictly positive (or bistochastic), to the best of our knowledge the same has not been shown to hold for rank non-decreasing channels. In \app{app:completely-rank-non-decreasing} (\propref{prop:complete-rank-non-dec}) we prove  that  $\ee \otimes \idch$ is indeed rank non-decreasing whenever $\ee$ is.

{\bf Full consistency with thermodynamics: }
As discussed recently in Ref. \cite{Mohammady2025}, a pre-condition for even beginning to interpret  the measurement process as a \emph{thermodynamic process} which may be subject to thermodynamic laws is that the full compound $\hs \otimes \ha$ must be thermodynamically closed, exchanging at most mechanical energy (work) with an external environment, but not heat. That is, the auxiliary system $\ha$ must be extended to include all degrees of freedom (thermal baths, etc.) that may exchange heat with $\hs$ and amongst each other, so that the full process is \emph{adiabatic}.  If we wish to consider the interaction channel $\ee$ as an independent, autonomous part of the process---that is, a process that is independent of the prior state preparation of both the system and the apparatus, as well as the subsequent measurement of the pointer observable---as is commonly assumed, either explicitly or implicitly,  then consistency with the second law demands that  $\ee$ must be  a bistochastic (e.g. unitary) channel \cite{Purves-2020}, since otherwise it could be used to construct a \emph{perpetuum mobile}. Since bistochastic channels are rank non-decreasing, then as long as the apparatus preparation $\xi$ is also strictly positive, then the process will be fully consistent with both the second, and the (strong) third laws. We remark that for a process to be consistent with the second law, then the \emph{objectification} mechanism with which the pointer observable $\Z$ obtains definite values,   modelled by a $\Z$-compatible instrument $\jj$ acting in $\ha$, must also be taken into account. But  we may always choose the L\"uders instrument $\jj^L_x(\bigcdot) \coloneq \sqrt{Z_x} \bigcdot \sqrt{Z_x}$ which yields a bistochastic channel $\jj^L_\xx$, so that the full process $(\idchsys \otimes \jj^L_\xx) \circ \ee$ is bistochastic whenever $\ee$ is, and hence the full process is consistent with the second law \cite{Minagawa2023a, Mohammady2025}. Since all $\Z$-compatible objectification instruments result in the same instrument $\ii$ acting in the system, for simplicity in this paper we  consider only the pointer observable and not the instrument that objectifies it.

%-------------------------------------------------------------
%-----Section: results----------------------------------------
%-------------------------------------------------------------

\section{Results}\label{sec:results}

\subsection{Operations}\label{sec:results-operations}

In this section, we shall characterise the individual operations that admit a thermodynamically consistent process. In analogy with \eq{eq:instrument-implementation}, we say that the tuple $(\ha, \xi, \ee, Z)$, where $Z$ is a \emph{single effect} on $\ha$, is a process for an operation $\Phi$ acting in $\hs$ if it holds that
\begin{align}\label{eq:operation-implementation}
    \Phi(\bigcdot) =  \tra[\onesys \otimes Z \  \ee(\bigcdot \otimes \xi)]  \, .
\end{align}
Note that if $Z = \oneapp$, corresponding to the case where no post-selection takes place, then $\Phi$ is a channel.  By \defref{defn:thermo-consistent-process}, we define the following sets of thermodynamically consistent operations:

\

\begin{definition}[Thermodynamically consistent operations]
\label{defn:thermo-consistent-operation}

\

\begin{enumerate}[(I)]

    \item   {\bf Operations consist with the weak third law: }  $\opset_\I(\hs)$ is defined as the set of operations acting in $\hs$ that admit a process $(\ha, \xi, \ee, Z)$,  as in \eq{eq:operation-implementation}, such that   $\xi$ is a strictly positive state and $\ee$ is a  strictly positive channel. 

    \item  {\bf Operations consistent with the strong third law: }  $\opset_\II(\hs)$ is defined as the set of operations acting in $\hs$ that admit a process $(\ha, \xi, \ee, Z)$,  as in \eq{eq:operation-implementation}, such that   $\xi$ is a strictly positive state and $\ee$ is a  rank non-decreasing channel. 

     \item  {\bf Operations fully consistent with thermodynamics: }  $\opset_\III(\hs)$ is defined as the set of operations acting in $\hs$ that admit a process $(\ha, \xi, \ee, Z)$,  as in \eq{eq:operation-implementation}, such that   $\xi$ is a strictly positive state and $\ee$ is a  bistochastic channel. 
\end{enumerate}
\unskip\nobreak\hfill $\square$
\end{definition}

 For any $C \in \{ \I, \II, \III\}$, $\opset_C(\hs)$ contains the identity channel $\idchsys$,  is convex,  and is closed under composition. That is, for any $\Phi_1, \Phi_2 \in \opset_C(\hs)$ and $0 \leqslant  \lambda \leqslant  1$, the operations $\Phi_3 (\bigcdot) \coloneq \lambda \, \Phi_1(\bigcdot) + (1-\lambda) \Phi_2(\bigcdot)$ and $\Phi_4 \coloneq \Phi_2 \circ \Phi_1$ admit a process that is subject to the thermodynamic constraint $C$, and are thus also members of $\opset_C(\hs)$. It follows that these sets form convex \emph{monoids}.  See \app{app:convexity} for the proof.

We also note that
\begin{align*}
&\opset_{\III}(\hs) \subsetneq \opset_{\II}(\hs) \subsetneq \opset_{\I}(\hs) \subsetneq \opset(\hs) \, .
\end{align*}

That each set in this chain is a subset of those appearing to its right follows trivially from \defref{defn:thermo-consistent-operation} and the fact that the set of bistochastic   channels  is a proper subset of the set of  rank non-decreasing channels, which is itself a   proper subset of the set of strictly positive channels, see \app{app:completely-rank-non-decreasing}. That each set is a  proper subsets of those appearing to its right  will be shown below in Theorems \ref{thm:weak-third}, \ref{thm:strong-third}, and \ref{thm:channel-III-sp-fixed-state}. Moreover, let us remark that since thermodynamic consistency does not restrict the pointer effect $Z$ in any way, then every effect $E$ admits an operation in $\opset_C(\hs)$ for every $C \in \{\I, \II, \III\}$. To see this,  consider a \emph{trivial} process where $\ha = \hs$, $\ee$ is a unitary swap channel,  $\xi >\zero$, and the pointer effect is chosen as  $Z = E$. This process implements the operation    $\Phi(\bigcdot) = \tr[E \, \bigcdot] \xi$, which is clearly compatible with $E$. Since such a process is  fully consistent with thermodynamics, then  $\Phi$ exists in $\opset_\III(\hs)$, and hence also in $\opset_\II(\hs)$ and $\opset_\I(\hs)$ by the above.

\ 

\begin{theorem}[Operations consistent with the weak third law] \label{thm:weak-third}
An operation $\Phi$ (that is compatible with a non-vanishing effect $E \ne \zero$) exists in  $\opset_\I(\hs)$ if and only if $\Phi$ is strictly positive. 
\unskip\nobreak\hfill $\square$ \end{theorem}
The necessity of strict positivity was shown already in Lemma D.1 of Ref. \cite{Mohammady2022a}.  The sufficiency is a new result;  see \app{app:weak-third-law} for further details and the proof. We remark that any effect $E\ne\zero$ admits a strictly positive $E$-compatible operation (\corref{app-cor:strict-positive-for-any-effect}). Moreover, it trivially follows from above that all rank non-decreasing operations, permitted only for strictly positive effects $E > \zero$, exist in $\opset_\I(\hs)$.

\

\begin{theorem}[Operations consistent with the strong third law]  \label{thm:strong-third}
A channel exists in $\opset_\II(\hs)$ if and only if it is rank non-decreasing, and any rank non-decreasing operation compatible with an indefinite effect $\zero < E < \onesys$ exists in $\opset_\II(\hs)$. 
\unskip\nobreak\hfill $\square$ \end{theorem}

We note that not all operations in $\opset_\II(\hs)$ are rank non-decreasing:  an operation is rank non-decreasing only if it is compatible with a strictly positive effect (\corref{app-cor:strict-positive-for-any-effect}), whereas every effect, including those with a non-trivial kernel, admit an operation in $\opset_\II(\hs)$. Moreover,  not all rank non-decreasing operations exist in $\opset_\II(\hs)$. This is  because every norm-1 and strictly positive effect admits a rank non-decreasing operation that has  non-vanishing fixed points, whereas  the only operations in $\opset_\II(\hs)$ that have non-vanishing fixed points are channels, i.e., operations compatible with the unit effect. To show  the latter claim, we note that an $E$-compatible operation $\Phi$ has non-vanishing fixed points only if $E$ is a norm-1 effect, and only if $\Phi$ has fixed states;  a state $\rho$ is a fixed point of $\Phi$ only if $\tr[E \rho] = \tr[\Phi(\rho)] = \tr[\rho] = 1$. But we obtain the following result:

\

\begin{lemma}\label{lemma:opII-rank-inc}
 Let $E$ be a non-trivial norm-1 effect, and let $\Phi \in \opset_\II(\hs)$ be an $E$-compatible operation. Then for all $\rho \in \s(\hs)$ the following holds:
 \begin{align*}
     \tr[E \rho] = 1 \implies \rank{\Phi(\rho)} > \rank{\rho} \, .
 \end{align*}
\unskip\nobreak\hfill $\square$ \end{lemma}

The above Lemma shows that for any $E$-compatible $\Phi \in \opset_\II(\hs)$,  unless $E = \onesys$, then $\Phi(\rho) \ne \rho $ for any state $\rho$. For example, consider  the L\"uders operation $\Phi^L(\bigcdot) \coloneq \sqrt{E} \bigcdot \sqrt{E}$. Such an operation is strictly positive, and rank non-decreasing (rank-preserving),  if and only if  $E$ is a strictly positive effect. Therefore,  all L\"uders operations compatible with  $E > \zero$ exist in $\opset_\I(\hs)$.  Now, if $E$ is a  norm-1 effect,   then the L\"uders operation does not disturb any state $\rho$ satisfying $\tr[E \rho] = 1$. As such,  for any strictly positive norm-1 effect  $\zero < E \leqslant \onesys$ and $E \ne \onesys$,  the L\"uders operation  exists in $\opset_\I(\hs)$ but not in $\opset_\II(\hs)$. On the other hand,    if $\zero < E < \onesys$ then the L\"uders operation does not have any non-vanishing fixed points and, as  discussed above, it also exists in   $\opset_\II(\hs)$.  See \app{app:strong-third-law} for further details and proof of the theorem and lemma.

\

\begin{theorem}[Channels that are fully consistent with thermodynamics ]\label{thm:channel-III-sp-fixed-state}
All channels in $\opset_\III(\hs)$ have a strictly positive fixed state.
\unskip\nobreak\hfill $\square$ \end{theorem}

See \app{app:strong-adbatic} (\propref{prop:Kraus-block-form}) for the proof. In particular, all bistochastic channels (such as unitary ones) which preserve the complete mixture exist in $\opset_\III(\hs)$. It follows from above that while all channels in $\opset_\III(\hs)$ are rank non-decreasing, not all rank non-decreasing channels exist in $\opset_\III(\hs)$, since there exists rank non-decreasing channels that perturb all strictly positive states.     For example, consider the channel $ \Phi(\bigcdot) = \lambda\, \idchsys (\,\bigcdot\,) + (1-\lambda)\, \tr[\,\bigcdot\,]\, \proj\phi$ where: $0< \lambda < 1$ ;   $\idchsys$ is the identity channel acting in $\hs$;  and $|\phi\>$ is a unit vector in $\hs$. It is easy to show that $\Phi(\sigma) = \lambda \,  \sigma + (1-\lambda) \proj\phi \geqslant \lambda \, \sigma $ for all states $\sigma$, which implies that $\rank{\Phi(\sigma)} \geqslant \rank{\sigma}$ for all $\sigma$. Therefore,  $\Phi$ is rank non-decreasing, and so it  exists in $\opset_\II(\hs)$. However, $\ff(\Phi) = \co \proj\phi$, and so it does not exist in $\opset_\III(\hs)$ \cite{Mohammady2024}.

The fact that any channel  $\Phi \in \opset_\III(\hs)$ has a strictly positive fixed state $\rho_0 = \Phi(\rho_0) > \zero$ guarantees that the fixed-point set of the dual channel,  $\ff(\Phi^*)$,  is a von Neumann algebra and, in particular, that it is closed under multiplication: for any  $A,B \in \ff(\Phi^*)$ it holds that $AB \in \ff(\Phi^*)$ \cite{Bratteli1998, Arias2002}. This has important consequences for non-disturbing measurements, discussed in the next section. 

Furthermore,  as recently shown in Ref. \cite{Mohammady2025}, a non-trivial effect does not admit a  purity-preserving operation in $\opset_\III(\hs)$. An operation is  purity-preserving when it maps pure inputs to pure outputs, and is completely purity preserving when this holds even when acting locally on an entangled bipartite system, in which case the operation is represented with a single Kraus operator.  That is, all operations in $\opset_\III(\hs)$ that are compatible with a non-trivial effect are represented with at least two Kraus operators, and they take at least some pure input state to a mixed output.  On the other hand, an effect $E$ admits a completely purity preserving operation in $\opset_\I(\hs)$ and $\opset_\II(\hs)$ if and only if $E > \zero$ and $\zero < E < \onesys$, respectively.  In particular, this implies that even for an indefinite effect $\zero < E < \onesys$, while the corresponding L\"uders operation $\Phi^L(\bigcdot) \coloneq \sqrt{E} \bigcdot \sqrt{E}$ does exist in $\opset_\II(\hs)$, it does not exist in $\opset_\III(\hs)$.

\

While the set of operations that are fully thermodynamically consistent, $\opset_\III(\hs)$,  is a proper subset of the set of all operations,  $\opset(\hs)$, we observe the following:

\

\begin{prop}
\label{rem:full-choi-rank}
All  operations  in the interior of $\opset(\hs)$, i.e., operations with a strictly positive Choi operator, exist in $\opset_\III(\hs)$. That is,  $\mathrm{int}\left(\opset(\hs)\right)\subsetneq\opset_\III(\hs)$.    
\unskip\nobreak\hfill $\square$ \end{prop}
\begin{proof}
    We start the proof by showing that for all $\Phi\in\opset(\hs)$ and all  $\epsilon >0$, there exists a $\Phi_1 \in\opset_\III(\hs)$ in the $\epsilon$-neighbourhood of   $\Phi$. By the Stinespring-Naimark-Ozawa dilation theorem, for any $\Phi \in \opset(\hs)$ there exists a process $(\ha, \proj0, \ee, Z)$ where  $\proj0$ is a pure state on $\ha$, $\ee$ is a unitary channel on $\hs \otimes \ha$,  and $Z$ is a projection on $\ha$,   so that
    \begin{align*}
        \Phi(\bigcdot) = \tra[(\onesys \otimes Z)\ \ee(\bigcdot \otimes \proj0)]\,. 
    \end{align*}
    Now consider the same process as above, but replace $\proj0$ with the state  $\xi=\frac{1}{1+\epsilon} \, \proj0 +\frac{\epsilon}{1+\epsilon} \, \Omega$, where $\Omega$ is a strictly positive state and  $\epsilon>0$.  The process $(\ha, \xi, \ee, Z)$   implements the operation  
    \begin{align*}
      \Phi_1(\bigcdot) &= \frac{1}{1+\epsilon}\,\tra [(\onesys \otimes Z)\ \ee( \bigcdot \otimes \proj0)] + \frac{\epsilon}{1+\epsilon} \, \tra [(\onesys \otimes Z)\ \ee( \bigcdot \otimes \Omega)] \\
      & = \frac{1}{1+\epsilon} \Phi(\bigcdot) + \frac{\epsilon}{1+\epsilon} \, \Phi_2(\bigcdot)
    \end{align*}
   where we define $\Phi_2(\bigcdot) \coloneq  \tra [(\onesys \otimes Z)\ \ee( \bigcdot \otimes \Omega)]$. Since $\xi > \zero$ and $\ee$ is bistochastic, $\Phi_1$ exists in $\opset_\III(\hs)$. But $\Phi - \Phi_1 = \epsilon (\Phi_1 - \Phi _2)$, and so  (for any topology induced by a metric, for example by the trace-norm) $\Phi_1$ is in the $\epsilon$-neighbourhood  of $\Phi$.  That $\mathrm{int}\left(\opset(\hs)\right)\subset\opset_\III(\hs)$ follows straightforwardly from the fact that $\opset_\III(\hs)$ is convex and therefore has no punctures. That $\mathrm{int}(\opset(\hs))$ is a proper subset of $\opset_\III(\hs)$ follows from the fact that some boundary points, such as unitary channels, exist in the latter but not in the former. 
\end{proof}

We note that operations with a strictly positive Choi operator map all input states to strictly positive output states. Such operations are clearly rank non-decreasing, and have a fixed state if and only if they are channels. Indeed, all fixed states of  channels with a strictly positive Choi operator are strictly positive.  The above shows that for any $C \in \{\I, \II, \III\}$, while $\opset_C(\hs) \subsetneq \opset(\hs)$, the closure of $\opset_C(\hs) $  is equal to $\opset(\hs)$. Additionally, since $\opset(\hs)$ is a compact convex set, this suggests that   $\opset_C(\hs)$  occupies the entire volume within $\opset(\hs)$.

\subsection{Instruments}\label{sec:results-instruments}

In this section, we shall characterise the instruments that admit a thermodynamically consistent process, and subsequently determine the (non)disturbance properties of instruments within each class.  By \defref{defn:thermo-consistent-process}, we define the following sets of thermodynamically consistent instruments:

\

\begin{definition}[Thermodynamically consistent instruments]
\label{defn:thermo-consistent-instrument}

\

\begin{enumerate}[(I)]

    \item   {\bf Instruments consistent with the weak third law: }  $\inset_\I(\hs)$ is defined as the set of instruments acting in $\hs$ that admit a process $(\ha, \xi, \ee, \Z)$,  as in \eq{eq:instrument-implementation}, such that   $\xi$ is a strictly positive state and $\ee$ is a  strictly positive channel. 

    \item  {\bf Instruments consistent with the strong third law: }   $\inset_\II(\hs)$ is defined as the set of instruments acting in $\hs$ that admit a process $(\ha, \xi, \ee, \Z)$,  as in \eq{eq:instrument-implementation}, such that   $\xi$ is a strictly positive state and $\ee$ is a  rank non-decreasing channel. 

     \item  {\bf Instruments 
     fully consistent with thermodynamics: }   $\inset_\III(\hs)$ is defined as the set of instruments acting in $\hs$ that admit a process $(\ha, \xi, \ee, \Z)$,  as in \eq{eq:instrument-implementation}, such that   $\xi$ is a strictly positive state and $\ee$ is a  bistochastic channel. 
\end{enumerate}
\unskip\nobreak\hfill $\square$
\end{definition}

It is clear that
\begin{align*}
&\inset_{\III}(\hs) \subsetneq \inset_{\II}(\hs) \subsetneq \inset_{\I}(\hs) \subsetneq \inset(\hs) \, ,
\end{align*}
and that for any $C \in \{\I, \II, \III\}$, an instrument $\ii\coloneq \{\ii_x : x \in \xx\}$ exists in $\inset_C(\hs)$ only if each operation $\ii_x$, as well as the channel $\ii_\xx(\bigcdot) \coloneq \sum_{x\in \xx} \ii_x(\bigcdot)$, exists in $\opset_C(\hs)$ defined in  \defref{defn:thermo-consistent-operation}. Indeed, by the discussion from the previous section, we immediately obtain the following:

\

\begin{corollary}\label{cor:instrument-classes}
Let $\ii\coloneq \{\ii_x : x \in \xx\}$ be an instrument acting in $\hs$, compatible with the observable $\E \coloneq \{E_x : x \in \xx\}$. The following hold:
    \begin{enumerate}[(i)]
        \item $\ii$ exists in $\inset_\I(\hs)$ if and only if each operation $\ii_x$ is strictly positive.
        \item $\ii$ exists in $\inset_\II(\hs)$ only if the channel $\ii_\xx$ is rank non-decreasing. On the other hand, if all operations $\ii_x$ are rank non-decreasing and compatible with indefinite effects $\zero < E_x < \onesys$,  then $\ii$ exists in $\inset_\II(\hs)$. 
        \item $\ii$ exists in $\inset_\III(\hs)$ only if the channel $\ii_\xx$ has a strictly positive fixed state, and only if each operation $\ii_x$ compatible with a non-trivial effect $E_x$ is not purity preserving, i.e., only if $\ii_x$ is represented by at least two Kraus operators, and maps some pure input state to a mixed output. 
    \end{enumerate}
\unskip\nobreak\hfill $\square$ \end{corollary}

As discussed in the previous section, since we are not restricted in our choice of pointer observable $\Z$, then thermodynamic   consistency  does not preclude the \emph{measurability} of any observable $\E$: For every observable $\E$, we may  construct a trivial measuring process utilising a unitary swap interaction channel  that implements the $\E$-compatible instrument $\ii_x(\bigcdot) = \tr[E_x \, \bigcdot] \xi$. But note that such an instrument is fully disturbing, since the posterior state $\xi$ contains no information at all about the prior state $\rho$. Therefore, we may now ask how consistency with thermodynamics limits the disturbance properties of measurements.

In what follows, we shall consider only instruments compatible with non-trivial observables $\E := \{E_x : x \in \xx\}$, i.e.,  observables  with more than one outcome, $N:= |\xx| >1$, and such that at least some effect in the range of $\E$ is not proportional to the identity. In other words, we consider only observables that provide information about the system to be measured. While informative measurements necessarily disturb at least some observable of the system being measured---no information without disturbance \cite{Busch2009}---some observables admit measurements that are non-disturbing in the sense that they do not disturb some property of the observable being measured.

\
\begin{figure}[tbp]
\centering
\includegraphics[width=0.9\textwidth]{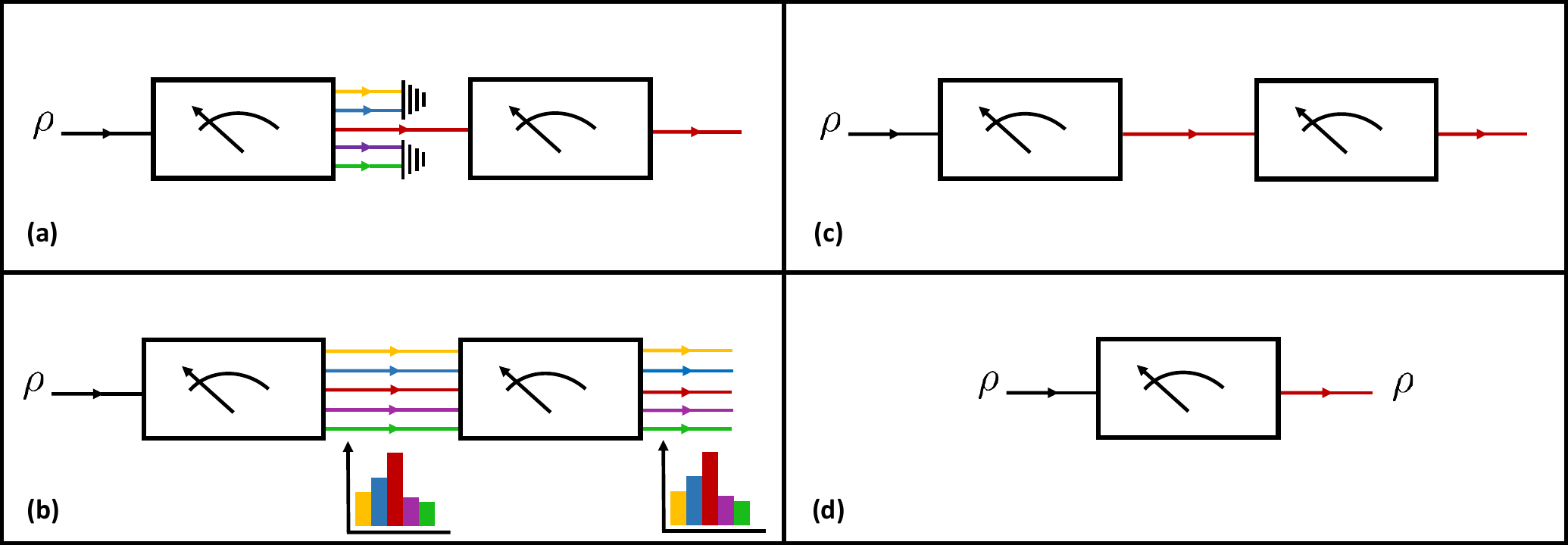}
\caption{{\bf Non-disturbing instruments (colour online).} In every frame (a)-(d), each box represents an instrument that measures the same observable, which takes a quantum system in an arbitrary state $\rho$ as an input from the left, and produces both a classical output (the measurement outcome) and a quantum output (the post-measurement state of the system) on the right. The different measurement outcomes and conditional post-measurement states are represented by the separate coloured arrows exiting from the right. {\bf (a):} The measurement is repeatable if, conditional on obtaining a given outcome in the first measurement (the red arrow that is allowed to enter the second instrument), the second measurement produces the same outcome with probability 1. {\bf (b):} The measurement is of the first kind if consecutive measurements produce the same statistics, represented here by identical histograms. {\bf (c): } The measurement is value reproducible if, whenever the first measurement produces a single outcome with probability 1 (only a single red arrow exits the instrument) the second measurement  produces the same outcome with probability 1. {\bf (d):} The measurement is ideal if, whenever the first measurement produces a single outcome with probability 1, the measurement does not disturb the state of the system (the output state is equal to the input state $\rho$).  For any observable, $(a) \implies (b) \implies (c)$ and $(d) \implies (c)$, but the converse implications do not hold in general. 
}\label{fig:measurements}
\end{figure}

\begin{definition}[Non-disturbing instruments]
An $\E$-compatible instrument $\ii:= \{\ii_x : x \in \xx\}$ acting in $\hs$ is called a
\begin{enumerate}[(a)]
 \item   repeatable measurement of $\E$ if $\E$ is a norm-1 observable, and if

 \begin{align*}
     \tr[E_y \, \ii_x(\rho)] = \delta_{x,y} \tr[E_x \, \rho] \qquad  \forall x,y \in \xx,  \rho \in \s(\hs) \, .
 \end{align*}
Equivalently, if $E_x \in \ff(\ii_x^*)$ for all $ x \in \xx$. In other words, if $\ii$ is a repeatable measurement of $\E$, then consecutive measurements of $\E$ by $\ii$ are guaranteed (with probability 1) to produce the same outcome. 
 \item   first-kind measurement of $\E$ if

  \begin{align*}
     \tr[E_x \, \ii_\xx(\rho)] = \tr[E_x \,  \rho] \qquad \forall x \in \xx,  \rho \in \s(\hs) \, .
 \end{align*}
Equivalently, if $\E := \{E_x : x\in \xx\} \subset \ff(\ii_\xx^*)$. In other words, if $\ii$ is a first-kind measurement of $\E$, then consecutive measurements of $\E$ by $\ii$ are guaranteed to produce the same statistics.

 \item  value-reproducible measurement of $\E$ if $\E$ is a norm-1 observable, and if

  \begin{align*}
     \tr[E_x \, \rho] = 1 \implies \tr[E_x \, \ii_\xx(\rho)] = 1 \qquad \forall x \in \xx,  \rho \in \s(\hs) \, .
 \end{align*}
In other words, if $\ii$ is a value reproducible measurement of $\E$, then if $\E$ has the value $x$ in any state $\rho$---if the outcome $x$ can be predicted to obtain with (probabilistic) certainty in this state---then $\E$ will continue to have value $x$ in the state  obtained after a ``non-selective'' measurement, i.e.,  $\ii_\xx(\rho)$.

 \item ideal measurement of $\E$ if $\E$ is a norm-1 observable, and if

\begin{align*}
     \tr[E_x \, \rho] = 1 \implies \ii_x(\rho) = \rho \qquad \forall x \in \xx,  \rho \in \s(\hs) \, .
\end{align*}
In other words, if $\ii$ is an ideal measurement of $\E$, then $\ii$ does not disturb the state of the measured system whenever an outcome can be predicted to obtain in this state with (probabilistic) certainty. 
 
\end{enumerate}
\unskip\nobreak\hfill $\square$
\end{definition}

\

A repeatable measurement is of the first kind, and a first-kind measurement (of a norm-1 observable) is value reproducible. While the converse implications do not  hold in general, in the case of sharp (projection valued) observables repeatability, first-kindness, and value-reproducibility coincide (see Theorem 10.3 in Ref. \cite{Busch2016a}). On the other hand, an ideal measurement is value reproducible, but a measurement may be value reproducible but not ideal. We note that for any norm-1 observable $\E$,  the corresponding  L\"uders instrument $\ii^L_x(\bigcdot) = \sqrt{E_x}\, \bigcdot\, \sqrt{E_x}$ is an ideal measurement of $\E$. In the case of sharp observables, the ideal measurements are precisely the L\"uders instruments, but unsharp norm-1 observables admit ideal measurements that are not of the L\"uders form.  Moreover, while ideal measurements of sharp observables are always repeatable, a repeatable measurement is guaranteed to be ideal only for the case of  rank-1 sharp observables, i.e., observables  all effects of which  are rank-1 projections; indeed, for sharp rank-1 observables we have $(a) \iff (b) \iff(c) \iff(d)$. Finally, let us remark that for every norm-1 observable there exists an instrument $\ii \in \inset(\hs)$ that may satisfy any of the properties (a)-(d), and while repeatability, value reproducibility, and ideality are  permitted only for norm-1 observables, an indefinite observable (which necessarily lacks the norm-1 property) may admit a measurement of the first kind. For example, the L\"uders instrument for a commutative observable, i.e., an observable all whose effects mutually commute, is a first-kind measurement.  This will be important in what follows.

\

In Ref. \cite{Mohammady2022a} it was shown that when an $\E$-compatible instrument $\ii$ exists in $\inset_\I(\hs)$, i.e., when the premeasurement interaction and apparatus preparation are strictly positive, then repeatability will be ruled out for all observables. Additionally, it was shown that if  the $\E$-channel $\ii_\xx$ has a strictly positive fixed state \cite{Mohammady2024},  then  ideality is also precluded for all observables, while first-kindness is permitted only for indefinite (or completely unsharp) commutative observables. While it may be the case that $\ff(\ii_\xx)$ does not have any strictly positive  states when $\ii \in \inset_\I(\hs)$, as we have seen this condition is guaranteed if $\ii  \in \inset_\III(\hs)$, i.e.,   if the premeasurement channel is not only strictly positive, but is bistochastic: all of the no-go results in Ref. \cite{Mohammady2022a} hold for instruments that are fully consistent with thermodynamics.  But as the following shows, even if   $\ii$ belongs to   $\inset_\II(\hs)$ but is not in $\inset_\III(\hs)$, i.e. if the premeasurement channel is rank non-decreasing but not bistochastic,  so that $\ff(\ii_\xx)$ need not contain any strictly positive states,  then nearly all of the no-go results in Ref. \cite{Mohammady2022a} will still carry over. This demonstrates that the strong third alone is responsible for the thermodynamic inconsistency of almost all types of non-disturbing measurements.  

\

\begin{theorem}\label{thm:non-ddisturbing-thermo-nogo}
Consider an $\E$-compatible instrument $\ii:= \{\ii_x : x\in \xx\}$ acting in $\hs$, and assume that  $\E$ is a non-trivial observable. Assume that $\ii$ belongs to $\inset_C(\hs)$ for  $C \in \{\I, \II, \III\}$ as given in \defref{defn:thermo-consistent-instrument}.  The following hold:
\begin{enumerate}[(i)]

    \item 

    If $\ii \in \inset_\I(\hs)$, then $\ii$ is not repeatable.

    \item 
    
    If   $\ii \in \inset_\I(\hs)$ and $\E$ is projection valued, then $\ii$ is not   first-kind, value reproducible, or ideal.

    \item 

    If $\ii \in \inset_\II(\hs)$, then $\ii$ is not ideal.

    \item 

    If $\ii \in \inset_\II(\hs)$, then $\ii$ is not value reproducible.

    \item 
    
    If $\ii \in \inset_\II(\hs)$, then $\ii$ is first-kind only if $\zero < E_x < \onesys $ for all $x \in \xx$.

    \item 

    If $\ii \in \inset_\III(\hs)$, then $\ii$ is first-kind only if $\zero < E_x < \onesys $ for all $x \in \xx$ and $[E_x, E_y] = \zero$ for all $x,y \in \xx$.

\end{enumerate}
\unskip\nobreak\hfill $\square$ \end{theorem}

\begin{proof}
\begin{enumerate}[(i)]
    \item 
    
    For any outcome $x$,  $E_x \ne \zero$  and $\ii_x$ is a strictly positive operation.  Therefore,  for any strictly positive state $\rho$ it holds that $\ii_x(\rho) > \zero$, and so  $\tr[E_y \ii_x(\rho)]   >0$,  for all $x,y$. As such, $\ii$ cannot be repeatable.

    \item 
    
    For projection valued observables, repeatability, first-kindness, and value reproducibility coincide. By item (i), these are all ruled out. Since ideality implies value reproducibility, then ideality is also ruled out.

    \item 
    
   By  \lemref{lemma:opII-rank-inc}, if  $\tr[E_x \rho] =1$ then $\rank{\ii_x(\rho)} > \rank{\rho}$, and so $\ii_x(\rho) \ne \rho$.  It follows that $\ii$ cannot be ideal.

    \item 
    
    If $\tr[E_x \rho] = 1$, then $\ii_y(\rho) = \zero$ for all $y \ne x$, which implies that $\ii_\xx(\rho) = \ii_x(\rho)$.  Recall that $\tr[E_x \rho]=1$ if and only if $\rho$ has support only in the eigenvalue-1 eigenspace of $E_x$. Let $P$ denote the projection onto this eigenspace, and consider a state  $\sigma$ such that $\supp(\sigma) = P \hs$. Given that $\tr[E_x \sigma] =  1$,   by  \lemref{lemma:opII-rank-inc} it follows that $\rank{\ii_\xx(\sigma)} = \rank{\ii_x(\sigma)} >\rank{\sigma}$. But this implies that  $ \supp(\ii_\xx(\sigma)) \not\subset  P \hs$, so that $\tr[E_x \ii_\xx (\sigma) ] < 1$. As such, $\ii$ cannot be value reproducible.

    \item (Sketch of the proof; for the full proof, see \app{app:fixed-point-measurement}) 
    
By the Schauder–Tychonoff fixed point theorem, $\ff(\ii_\xx)$ contains at least one state. Define $\P$ as the minimal support projection on $\ff(\ii_\xx)$, i.e., for any projection $Q$ such that $Q \rho = \rho \, \,  \forall \rho \in \ff(\ii_\xx)$, it holds that $\P \leqslant Q$. If $\ff(\ii_\xx)$ contains a strictly positive state, then $\P = \onesys$, in which case the statement follows from  Theorem 4.2 of Ref. \cite{Mohammady2022a}. For every effect $E_x$ we may write 
    \begin{align*}
        \P E_x \P &=  \bigoplus_{\alpha}\lambda_\alpha(x) P_\alpha,
    \end{align*}
   where $P_\alpha$ are mutually orthogonal projections such that $\sum_\alpha P_\alpha = \P$,  and  $0 < \lambda_{\alpha}(x) < 1$. There exists a unital CP map  $\ii_\av^*(\bigcdot) = \ii_\av^*( \P \bigcdot \P)$ such that $E_x = \ii_\xx^*(E_x) \iff E_x = \ii_\av^*(E_x)$.  It follows that if $\ii$ is a measurement of the first kind, then for all $x$ it must hold that 
\begin{align*}
    \| E_x\| = \| \ii_\av^*(E_x) \| = \| \ii_\av^*( \P E_x \P) \| \leqslant \| \P E_x \P \| <1 \, ,
\end{align*}
where the first inequality follows from the fact that $\ii_\av^*$ is CP and unital, and the final inequality follows from the fact that $\lambda_{\alpha}(x) <1$. Similarly, we may write
\begin{align*}
    \| \onesys -  E_x\| & = \| \ii_\av^*(\onesys - E_x) \| = \| \ii_\av^*(\P - \P E_x \P) \|  \leqslant \| \P - \P E_x \P \| <1 \, ,
\end{align*}
where the final inequality follows from the fact that $\lambda_{\alpha}(x) >0$. It follows that $E_x$ cannot have eigenvalue 1 or eigenvalue 0, and so $\E$ must be indefinite, i.e., $\zero < E_x < \onesys$. 

\item The requirement that $\zero < E_x < \onesys$ must hold follows immediately from (v) and the fact that $\inset_\III(\hs) \subset \inset_\II(\hs)$. The requirement that $\E$ must be  commutative, i.e., $[E_x, E_y]= \zero$, follows from  \thmref{thm:channel-III-sp-fixed-state} which states that if $\ii \in \inset_\III(\hs)$ then $\ff(\ii_\xx)$ contains a strictly positive state, and so $\ff(\ii_\xx^*)$ is a von Neumann algebra, and Proposition 4 of Ref. \cite{Heinosaari2010}.

\end{enumerate}

\end{proof}

\

\begin{remark}
An $\E$-compatible instrument $\ii \in \inset_\I(\hs)$ cannot be repeatable for any observable $\E$. While ideality, value reproducibility, and first-kindness are also precluded for sharp (projection valued) observables, these may be admitted for observables that are unsharp, but with the norm-1 property so that they are not indefinite.  We may show this using the following example. Let the system be $\hs = \co^3$ with orthonormal basis $\{|\pm\>, |0\>\}$.  Consider the binary norm-1 observable $\E \coloneq \{E_+, E_-\}$  on $\hs$, with value space $\xx = \{+,-\}$, defined by $E_\pm  \coloneq \proj{\pm} + \frac{1}{2} \proj{0}$. Consider the $\E$-compatible instrument $\ii$ with operations 
 \begin{align*}
     \ii_\pm(\bigcdot) \coloneq \<\pm| \bigcdot |\pm\> \proj{\pm} + \<0| \bigcdot |0\> \frac{\onesys}{6}. 
 \end{align*}
 It is easily verified that $\ii_\pm$ are strictly positive operations, and so  $\ii$ exists in $\inset_\I(\hs)$. It is clear that this measurement is ideal, since $\tr[E_\pm \rho ] =1 \iff \rho = \proj{\pm}$, and $\ii_\pm(\proj{\pm}) = \proj{\pm}$.  Indeed, since ideality implies value reproducibility, then $\inset_\I(\hs)$ also contains value reproducible measurements. Finally, note that 
 \begin{align*}
   &\ii_\xx(\bigcdot)  = \sum_{a = \pm} \<a| \bigcdot |a\> \proj{a}  + \<0| \bigcdot |0\> \frac{\onesys}{3} \, ,
     &\ii_\xx^*(\bigcdot) =  \sum_{a = \pm} \<a| \bigcdot |a\> \proj{a} + \frac{1}{3}\tr[\bigcdot ] \proj{0} \, .
 \end{align*}
 It is easily verified that $\ii_\xx^*(E_\pm ) = \proj{\pm} + \frac{1}{3}\tr[E_\pm ] \proj{0} = \proj{\pm} + \frac{1}{2} \proj{0} = E_\pm$, and so this measurement is also first-kind. 
 But recall that any $\ii \in \inset_\I(\hs)$ for which the channel $\ii_\xx$ has a strictly positive fixed state cannot be ideal or value reproducible or, if $\E$ is not indefinite, of the  first kind.  This does not contradict what we observed, since for the instrument defined above,  $\rho = \ii_\xx(\rho)$ only if $\<0| \rho|0\> = 0$. That is, $\ii_\xx$ perturbs all strictly positive states.
\unskip\nobreak\hfill $\square$ \end{remark}

\

\begin{remark}
Consider an $\E$-compatible instrument $\ii \in \inset_\II(\hs)$, and assume that for some outcome $x$, the effect $E_x$ has rank 1, i.e., $E_x = \lambda \proj{\psi}$ for some unit vector $|\psi\>$ in $\hs$.   It follows that $\ff(\ii_\xx^*) = \co \onesys$. That is, $\ii$ disturbs all non-trivial observables. See \app{app:fixed-point-measurement} (\corref{cor:rank-1-disturbance}) for the proof. 
\unskip\nobreak\hfill $\square$ \end{remark}

\

\begin{remark}
For every observable $\E$ that is both indefinite and commutative, there exists a corresponding instrument $\ii$ in $\inset_\II(\hs)$ that is a measurement of the first kind; for every indefinite observable the corresponding L\"uders instrument $\ii^L_x(\bigcdot) = \sqrt{E_x} \bigcdot \sqrt{E_x}$ exists in  $\inset_\II(\hs)$, since every operation of this instrument is rank non-decreasing,   and the L\"uders instrument is a measurement of the first kind if and only if the corresponding observable is commutative. However, recall that the L\"uders instrument does not exist in $\inset_\III(\hs)$, since the operations of such instruments are completely purity-preserving, i.e., represented with a single Kraus operator \cite{Mohammady2025}. Therefore, while indefiniteness and commutativity of an observable are necessary for the existence of a first-kind measurement in $\inset_\III(\hs)$, these are not sufficient. Notwithstanding, there do exist some indefinite and commutative observables which admit a first-kind measurement in $\inset_\III(\hs)$. See Example G.1 in Ref. \cite{Mohammady2022a}, where a specific class of indefinite commutative observables admit a first-kind measurement utilising a strictly positive apparatus preparation and a unitary premeasurement interaction. 
\unskip\nobreak\hfill $\square$ \end{remark}

\

\begin{remark}
  Consider an $\E$-compatible instrument $\ii \in \inset_\III(\hs)$.  Assume that the bistochastic premeasurement interaction $\ee$ used in the implementation of $\ii$ also conserves some additive quantity $H \coloneq \hsys\otimes \oneapp + \onesys \otimes \happ$, where $\hsys \in \lo(\hs)$ and $\happ \in \lo(\ha)$ are self-adjoint operators representing the system and apparatus part of the conserved quantity, respectively.  That is,  $\ee^*(H) = H$. For example, $H$ can be the Hamiltonian, in which case the adiabatic implementation of the premeasurement interaction $\ee$ does not consume any work.  Then $\ii$ is a first-kind measurement only if it additionally holds that $[E_x, \hsys] = \zero$ for all $x \in \xx$. This follows from the fact that $\ii \in  \inset_\III(\hs)$ implies that $\ff(\ii_\xx)$ contains a strictly positive state, which guarantees that $\ff(\ii_\xx^*)$ is a von Neumann algebra, and Theorem 4.1 in Ref. \cite{Mohammady2021a}.
\unskip\nobreak\hfill $\square$ \end{remark}

%-------------------------------------------------------------
%-----Section: discussion-------------------------------------
%-------------------------------------------------------------

\section{Discussion}\label{sec:discussion}

This work generalizes and unifies previous works relating to the thermodynamic consistency of quantum operations and measurements. In the conventional framework,  thermodynamically consistent operations---interpreted as operations consuming only finite thermodynamic resources---are considered as those that are implementable via a unitary interaction with an apparatus prepared in a strictly positive, i.e, full-rank, state. However, unitarity of the interaction between system and apparatus simultaneously satisfies several properties; unitary channels are  strictly positive, rank non-decreasing, and bistochastic, properties which (together with a strictly positive apparatus preparation) we identify with ($\I$) the weak third law, ($\II$) the strong third law, and ($\III$) the conjunction of the  second and third laws, respectively.  Therefore, to illuminate  what particular thermodynamic law is responsible for the (un)attainability of a given operation or measurement, we have  introduced the hierarchy of operations and instruments that are  ($\I$) consistent with the weak third law, ($\II$) consistent with the strong third law, and ($\III$) consistent with the second and the third laws, i.e.,  fully consistent with thermodynamics. Note that here, we are considering the possibly non-unitary interaction channels as fundamental objects that are not themselves dilated, so as to avoid issues of infinite regress.  

Each  class in the hierarchy was systematically analysed, with necessary (and in some cases also sufficient) conditions provided for an operation or measurement to belong to the class.   For example, in the case of quantum channels we saw that consistency with the weak and strong third laws is equivalent to the channel being strictly positive and rank non-decreasing, respectively. On the other hand, a channel is fully consistent with thermodynamics only if it is rank non-decreasing \emph{and} does not perturb some strictly positive state; the latter condition can be seen to be a generalisation of a key property of thermal channels, which do not perturb the thermal equilibrium state of the system. In the case of the non-disturbance properties of quantum measurements, we saw that while repeatability is forbidden by the weak third law,  ideality and value reproducibility are forbidden by the strong third law; that is, while some unsharp observables admit ideal or value reproducible measurements that are consistent with the weak third law, no observable admits such measurements in a way that is consistent with the strong third law. On the other hand, while first-kindness demands indefiniteness of the measured observable given the strong third law, that such an observable must also be commutative is necessitated only when the second law is also required to hold.

While we have fully characterised the set of operations consistent with the weak third law, in the sense that we provided both necessary and sufficient conditions for an operation to belong to this class, the other two classes in the hierarchy were not fully characterised: necessary and \emph{partially} sufficient conditions were provided for these. Furthermore, we did not explore whether or not operations that are fully consistent with thermodynamics differ with those that are implementable via a unitary interaction with a strictly positive apparatus preparation. Additionally, we only addressed the question of thermodynamic consistency of a given operation, i.e., implementability of said operation given resources that are finite, albeit arbitrarily large. A physically relevant question is how to quantify the resources that are required for the implementation of a given thermodynamically consistent operation, for example, by means of quantitative trade-off relations. We leave these open problems for future work.

%-------------------------------------------------------------
%-----Acknowledgments-----------------------------------------
%-------------------------------------------------------------
\begin{acknowledgments}
Funding for this project was provided by the IMPULZ program of the Slovak Academy of Sciences under the Agreement on the Provision of Funds No. IM-2023-79 (OPQUT). F. S. also acknowledges funding from project  09I03-03-V04-00777 (QENTAPP). M. H. M. also acknowledges funding from projects  VEGA 2/0164/25 (QUAS) and APVV-22-0570 (DeQHOST).     
\end{acknowledgments}

\appendix

%-------------------------------------------------------------
%-----Appendix: definite-effects------------------------------
%-------------------------------------------------------------

\section{Effects and fixed points of operations}\label{app:definite-effects}
A positive operator  $\zero \leqslant E \leqslant \one$ is called an effect. For any  effect $E$, there exists a  complementary effect $E^c \coloneq \one - E$. 

\

\begin{lemma}\label{lemma:effect-support-probability-1}
    Let $E$ be an effect on a finite-dimensional system $\h$. The following hold: 
    \begin{enumerate}[(i)]
        \item There exists a state $\rho$ such that $\tr[E \rho] =1$ if and only if $\|E\|=1$. 
        \item A state $\rho$ satisfies $\tr[E \rho] =1$ if and only if $E \rho = E \rho E = \rho$. 
        \item A state $\rho$ satisfies $\tr[E \rho] = 1$ if and only if $P \rho = P \rho P = \rho$, where $P$ is the projection onto the eigenvalue-1 eigenspace of $E$. 
    \end{enumerate}
\unskip\nobreak\hfill $\square$ \end{lemma}
\begin{proof}
\begin{enumerate}[(i)]
    \item The if statement is trivial, so we shall prove the only if statement. For any self-adjoint  $A = A^* \in \lo(\h)$, and any $B\in \lo(\h)$, it holds that $B^* A B \leqslant \|A\| B^*B$. It follows that for any state $\rho$, it holds that  $\tr[E \rho] = \tr[\sqrt{\rho} E \sqrt{\rho}] \leqslant \| E\| \tr[\rho] = \| E\|$. Since $E$ is an effect, then $\| E\| \leqslant 1$. As such, $\tr[E\rho]=1 \implies \|E\|=1$.  

    \item The if statement is trivial, so we shall  prove the only if statement. Assume that $\tr[E \rho] = 1$, which implies that  $\tr[E^c \rho] = 0$. But $\tr[E^c \rho] = \tr[(\sqrt{E^c} \sqrt{\rho})^*(\sqrt{E^c} \sqrt{\rho})]$, which vanishes if and only if   $\sqrt{E^c} \sqrt{\rho} =\zero \implies E^c \rho =\zero$, which gives $\rho = (E + E^c)\rho = E\rho$.  Since $\rho$ and $E$ are self-adjoint, we also have $\rho E = \rho$, and so $E \rho E = E \rho = \rho$. 

 \item We may decompose $E$ as $E = P + Q$, where $Q$ is a positive operator with orthogonal support to $P$, and which satisfies $\| Q\| <1$. Since $EP = P$, that $\tr[E\rho] =1$ if $\rho = P \rho$ immediately follows. Now note that $E^n = P + Q^n$ for any $n\in \nat$. But by (ii), if $\tr[E \rho] = 1$ then it must hold that $E^n \rho = P\rho + Q^n \rho = \rho$ for all $n$. Since $\|Q\|<1$ implies that $\lim_{n\to \infty } Q^n = \zero$, it follows that $P \rho = \rho$. Similarly as in (ii), this implies that $P \rho P = \rho$.
\end{enumerate}
\end{proof}

Recall that an operation that is compatible with the unit effect is a channel. By the Schauder–Tychonoff fixed point theorem \cite{Fixed-points-appl, Tumulka2024}, all channels mapping a system to itself have at least one fixed state. However, there exist operations that are not channels  which nonetheless have non-vanishing fixed points, but only if such operation is compatible with a norm-1 effect. 

\

\begin{lemma}\label{lemma:null-fixed-point}
Let $\Phi : \lo(\h) \to \lo(\h)$ be an $E$-compatible operation, and $\Phi^*$ its dual. The following hold:
 \begin{enumerate}[(i)]
     \item If $\| E \| <1$, then $\ff(\Phi) = \ff(\Phi^*) = \zero$.

     \item  $\ff(\Phi), \ff(\Phi^*)$ contain non-vanishing operators if and only if there exists a projection $P$ such that: (a) $EP = P$, and (b) the operation $\Phi_P(\bigcdot) \coloneq P \Phi( P \bigcdot P) P$ satisfies $\Phi_P(\bigcdot) = \Phi(P \bigcdot P)$.    
 \end{enumerate}
\unskip\nobreak\hfill $\square$ 
\end{lemma}

\begin{proof}
\emph{(i)}:  By complete positivity,  it trivially holds that $\Phi(\zero) = \Phi^*(\zero) = \zero$. If $\ff(\Phi)$ contains a non-vanishing fixed point, then it must contain a fixed state \cite[Theorem 6.5]{Wolf2012}. Assume that $\Phi(\rho_0) = \rho_0$ for some state $\rho_0$. This implies that 
\begin{align}
\label{appeq:fixed-state}
    \tr[E \rho_0] = \tr[\Phi(\rho_0)] = \tr[\rho_0] = 1 \, ,
\end{align}
which, by item (i) of \lemref{lemma:effect-support-probability-1},  implies that $\| E \| = 1$ must hold. Therefore, if $\| E \| <1$, then $\ff(\Phi) = \zero$. Since $\dim(\ff(\Phi)) = \dim(\ff(\Phi^*))$, then $\ff(\Phi^*) = \zero$ also holds. 

\

\emph{(ii)}:   To prove the only if statement, we shall show that if $\Phi$ has a fixed point, then the projection $P$ with the stated properties (a)-(b) exists. To this end, we first note that, as mentioned in item (i), if $\Phi$ has any non-vanishing fixed points, it must also have at least one fixed state  $\rho_0$ satisfying Eq.~\eqref{appeq:fixed-state}. Let  $P$ be the support projection for this state. 
By \lemref{lemma:effect-support-probability-1}, this implies that  $\|E\|=1$ and $EP=P$. That is, $P \leqslant \tilde P$, where  $\tilde P$ is the projection onto the eigenvalue-1 eigenspace of $E$. We thus have (a). Now, recall that for any  state $\sigma$ satisfying $P\sigma=\sigma$, there exists $\lambda>0$ such that $\rho_0 \geqslant \lambda\sigma$. By positivity and linearity of $\Phi$, it follows that $\rho_0=\Phi(\rho_0) \geqslant \lambda\Phi(\sigma)$ and so $P\Phi(\sigma)=\Phi(\sigma)$. That is to say, for any state $\sigma$ it holds that  $\supp(\sigma) \subseteq P \h \implies \supp(\Phi(\sigma)) \subseteq P \h $, and so     $\Phi(PAP)=P\Phi(PAP)P =: \Phi_P(A)$ for any $A$. We thus have (b).

Now we shall show the if statement. Assume that $P$ exists satisfying conditions (a)-(b). Recall that an operation is trace preserving, i.e., is a channel, if and only if its dual is unital. Observe that $\Phi_P(\bigcdot) = \Phi( P \bigcdot P) \iff \Phi_P^*(\bigcdot) = P \Phi^*(\bigcdot) P$, and recall that $\Phi^*(\one) = E$. It clearly follows that 
\begin{align*}
    \Phi_P^*(\one) =  P \Phi^*(\one) P = P E P = P \, . 
\end{align*}
Unless $P = \one$ then $\Phi_P$ is not a channel acting in $\h$. But, when we restrict $\Phi_P$ from $\lo(\h) \to \lo(\h)$ to $\lo(P \h) \to \lo(P\h)$, and note that the unit in $P\h$ is $P$, we see that $\Phi_P^*(P) = \Phi_P^*(\one) = P$. It follows that the restricted   $\Phi_P$ is a channel. By the Schauder–Tychonoff fixed point theorem  there exists at least one state $\rho_0 = P \rho_0$ such that $\Phi_P(\rho_0) = \rho_0$. But, this implies that $\rho_0 = \Phi_P(\rho_0) = \Phi(P \rho_0 P) = \Phi(\rho_0)$.

\end{proof}

%-------------------------------------------------------------
%-----Appendix: completely-rank-non-decreasing----------------
%-------------------------------------------------------------

\section{Strictly positive and rank non-decreasing CP maps}\label{app:completely-rank-non-decreasing}

Let $\Phi : \lo(\h) \to \lo(\kk)$ be a CP map. $\Phi$ is strictly positive if $A > \zero \implies \Phi(A) >\zero$. On the other hand, if $\kk = \h$, then $\Phi$ is  rank non-decreasing if $\rank{\Phi(\A)} \geqslant \rank{\A}$ for all  $A \geqslant \zero$. The composition of two strictly positive (or rank non-decreasing) CP maps is also strictly positive (rank non-decreasing).  While a rank non-decreasing CP map is clearly strictly positive, there exist strictly positive CP maps acting in $\h$ that are not rank non-decreasing \cite[Appendix B]{Mohammady2022a}. For any pair of operators $C_1, C_2 \in \lo(\h)$, we define   the $(C_1, C_2)$-operator scaling of $\Phi : \lo(\h) \to \lo(\h)$, and its dual,  as 
 \begin{align}\label{eq:operator-scaling}
  &\Phi_{C_1,C_2} ( \bigcdot )  \coloneq C_1 \Phi( C_2\, \bigcdot\, C_2 ^* ) C_1^* \, ,   &[\Phi_{C_1,C_2}]^* (\bigcdot) & = C_2^* \Phi^*( C_1^*\, \bigcdot\, C_1  ) C_2 \, .
 \end{align}
These are clearly both CP maps acting in $\h$.  Further, we define
\begin{align}\label{eq:DS-scaling}
 DS(\Phi_{C_1, C_2}) \coloneq \tr[(\Phi_{C_1,C_2}(\one\sub{\h}) - \one\sub{\h})^2] + \tr[([\Phi_{C_1,C_2}]^*(\one\sub{\h}) - \one\sub{\h})^2] \, .   
\end{align}
Now we recall a useful result, shown in Theorem 4.6 of Ref. \cite{Gurvits2003}.

\

\begin{lemma}\label{lemma:iff-rank-non-dec}
 Let $\Phi$ be a CP map acting in a finite dimensional system $\h$. $\Phi$ is rank non-decreasing if and only if for all $\epsilon >0$, there exists $C_1, C_2 \in \lo(\h)$ such that $DS(\Phi_{C_1, C_2}) \leqslant \epsilon^2$.
\unskip\nobreak\hfill $\square$ \end{lemma}
As an immediate corollary, we see that a bistochastic channel is rank non-decreasing; if $\Phi(\one) = \Phi^*(\one) = \one$, then $DS(\Phi_{\one,\one}) = 0$.

\

\begin{lemma}\label{lemma:equivalent-dual-positive-rank-non-dec}
    Let $\Phi : \lo(\h) \to \lo(\kk)$ be a CP map. The following hold:
\begin{enumerate}[(i)]
    \item  The following statements are equivalent: (a) $\Phi$ is strictly positive; (b) there exists $   \lo(\h) \ni  B \geqslant \zero $ such that $\Phi(B) > \zero$; (c) for all $A \in \lo(\kk)$ it holds that $\Phi^*(A^*A) = \zero \iff A = \zero$.
    \item If $\kk = \h$, then  $\Phi$ is rank non-decreasing if and only if $\Phi^*$ is rank non-decreasing.
\end{enumerate}    
\unskip\nobreak\hfill $\square$ \end{lemma}
\begin{proof}
$(i): $  (a) $\implies$ (b) is trivial. To show (b) $\implies$ (a), let us note that for any $A > \zero$ and $B \geqslant \zero$ on $\h$, there exists $\lambda > 0$ such that $A \geqslant  \lambda B$. Assume that $\Phi(B) > \zero$ for some $B \geqslant \zero$. It follows from positivity and linearity of $\Phi$ that $\Phi(A) \geqslant  \lambda \Phi(B) > \zero$ for all $A>\zero$. 
Now let us show that (a) $\implies$ (c). For any $\rho > \zero$ on $\h$, it holds that $\tr[\rho \, \Phi^*(A^*A) ] = 0 \iff \Phi^*(A^*A) = \zero$. Assume that $\Phi$ is strictly positive, so that for any $\rho > \zero$, we have that $\Phi(\rho) > \zero$. It follows that $\tr[\rho \, \Phi^*(A^*A)] = \tr[\Phi(\rho) A^*A] = 0 \iff A = \zero$.  As such,   $\Phi^*(A^*A) = \zero \iff A = \zero$.  
Now we shall show (c) $\implies$ (a).  Since $\Phi^*$ is a positive map, $\Phi^*(A^*A) \geqslant \zero$. Assume that $\Phi^*(A^*A) = \zero \iff A = \zero$.  For any strictly positive $\rho$ on $\h$ it holds that $\tr[\Phi(\rho) A^*A] = \tr[\rho \Phi^*(A^*A)] = 0 \iff A = \zero$, which implies that $\Phi(\rho) >\zero$, and so $\Phi$ is strictly positive.  

$(ii): $ By \eq{eq:operator-scaling}, we observe that $\Phi^*_{C_2^*, C_1^*} = [\Phi_{C_1, C_2}]^*$ and $[\Phi^*_{C_2^*, C_1^*}]^* = \Phi_{C_1, C_2}$, and so by \eq{eq:DS-scaling} it holds that $DS(\Phi^*_{C_2^*, C_1^*}) = DS(\Phi_{C_1, C_2})$. The statement follows trivially from \lemref{lemma:iff-rank-non-dec}.
\end{proof}

 Note that while the dual of a rank non-decreasing operation is rank non-decreasing, the dual of a strictly positive operation need not be strictly positive. This has the following implication:

\

\begin{corollary}
\label{app-cor:strict-positive-for-any-effect}
    For any effect $E \ne \zero$, there exists an $E$-compatible strictly positive operation. On the other hand, an $E$-compatible operation is rank non-decreasing only if $E > \zero$. But if $E > \zero$, then there exists an $E$-compatible rank non-decreasing operation. 
\unskip\nobreak\hfill $\square$ \end{corollary}
\begin{proof}
  For any effect $E \ne \zero$, the operation $\Phi(\bigcdot) \coloneq \tr[E\,  \bigcdot] \,\sigma$ is compatible with $E$, and is strictly positive when $\sigma > \zero$.   Note that $\Phi^*(\bigcdot) = \tr[\sigma \, \bigcdot] \, E$, which is not strictly positive unless $E$ is. On the other hand, $\Phi$ is rank non-decreasing if and only if $\Phi^*$ is. As such, if $\Phi$ is rank non-decreasing, then $E = \Phi^*(\one) > \zero$. Indeed, we see that if both $E$ and $\sigma$ are strictly positive, then both   $\Phi(\bigcdot) \coloneq \tr[E\,  \bigcdot] \,\sigma$ and $\Phi^*(\bigcdot) = \tr[\sigma \, \bigcdot] \, E$ are rank non-decreasing. 
\end{proof}
\begin{remark}
    Note that while every rank non-decreasing operation must be compatible with a strictly positive effect, every strictly positive effect  admits an operation that is not strictly positive, let alone rank non-decreasing; consider $\Phi(\bigcdot) = \tr[E\, \bigcdot ]\, \sigma$, which is compatible with $E > \zero$ but is strictly positive (and rank non-decreasing) only if $\sigma > \zero$.
\unskip\nobreak\hfill $\square$ \end{remark}

It is trivial that if $\Phi$ is a bistochastic channel acting in $\h$, then the extension $\Phi \otimes \idch$ where $\idch$ is the  identity channel acting in some space $\rr$ is also bistochastic. Similarly, if  $\Phi : \lo(\h) \to \lo(\kk)$ is a strictly positive CP map, then the extension $\Phi \otimes \idch$ is strictly positive,  which follows from the fact that $\Phi \otimes \idch\left(\one\sub{\h}\otimes\one\sub{\rr}\right)=\Phi(\one\sub{\h})\otimes\one\sub{\rr} > \zero$ if $\Phi(\one\sub{\h})> \zero$.  In other words, bistochastic channels are completely bistochastic, and strictly positive CP maps are completely strictly positive.  Now we shall show that the same property holds for rank non-decreasing operations.

\

\begin{prop}\label{prop:complete-rank-non-dec}
    Let  $\Phi$ be a rank non-decreasing CP map acting in a finite dimensional system $\h$. For any finite dimensional system $\rr$,  and $\idch$ the identity channel acting in  $\rr$,   $\Phi \otimes \idch$ is a rank non-decreasing CP map acting in $\h \otimes \rr$. 
\unskip\nobreak\hfill $\square$ \end{prop}
\begin{proof}
By \lemref{lemma:iff-rank-non-dec},  $\Phi \otimes \idch$ is rank non-decreasing if and only if for all $\epsilon > 0$, there exists $D_1, D_2 \in \lo(\h \otimes \rr)$ such that $DS((\Phi \otimes \idch)_{D_1, D_2}) \leq \epsilon^2$, with the operator scaling $(\Phi \otimes \idch)_{D_1, D_2}$ defined as \eq{eq:operator-scaling}. Let us define 
\begin{align*}
    &D_1 \coloneq C_1  \otimes \one\sub{\rr} \, , &D_2 \coloneq C_2 \otimes \one\sub{\rr} \, .
\end{align*}
We observe that  
\begin{align*}
(\Phi \otimes \idch)_{D_1, D_2}(\one\sub{\h} \otimes \one\sub{\rr}) & =   \Phi_{C_1,C_2}(\one\sub{\h}) \otimes \one\sub{\rr} \, ,
\end{align*}
and similarly 
\begin{align*}
[(\Phi \otimes \idch)_{D_1, D_2}]^*(\one\sub{\h} \otimes \one\sub{\rr}) & =   [\Phi_{C_1,C_2}]^*(\one\sub{\h}) \otimes \one\sub{\rr}\, .
\end{align*}
It is easy to verify that 
\begin{align*}
DS\left((\Phi \otimes \idch)_{D_1, D_2}\right) & = \tr[\left((\Phi \otimes \idch)_{D_1, D_2}(\one\sub{\h} \otimes \one\sub{\rr}) - \one\sub{\h} \otimes \one\sub{\rr}\right)^2] + \tr[\left([(\Phi \otimes \idch)_{D_1, D_2}]^*(\one\sub{\h} \otimes \one\sub{\rr}) - \one\sub{\h} \otimes \one\sub{\rr}\right)^2] \\
& = \tr[\left(\Phi_{C_1,C_2}(\one\sub{\h}) - \one\sub{\h}\right)^2 \otimes \one\sub{\rr}] + \tr[( [\Phi_{C_1,C_2}]^*(\one\sub{\h}) - \one\sub{\h})^2 \otimes \one\sub{\rr}] \\
& = \dim(\rr)\, \tr[(\Phi_{C_1,C_2}(\one\sub{\h}) - \one\sub{\h})^2 ] + \dim(\rr)\,\tr[( [\Phi_{C_1,C_2}]^*(\one\sub{\h}) - \one\sub{\h}))^2] \\
& =  \dim(\rr)\, DS(\Phi_{C_1,C_2}) .
\end{align*}
By assumption, $\Phi$ is rank non-decreasing. Therefore, by \lemref{lemma:iff-rank-non-dec}, for every $\epsilon >0$ we may choose  $C_1, C_2 \in \lo(\h)$ so that  $DS(\Phi_{C_1,C_2}) \leqslant \epsilon^2 / \dim(\rr)$. In such a case, we have $DS((\Phi \otimes \idch)_{D_1, D_2}) \leqslant \epsilon^2$, and so $\Phi \otimes \idch$ is rank non-decreasing. 
\end{proof}

\

\begin{lemma}\label{lemma:conditional-projection-rank-non-dec}
    Let $\{\Phi^{(i)}\}$ be rank non-decreasing CP maps acting in $\h$. Then 
    \begin{align*}
        \Lambda(  \bigcdot\sub{\h} \otimes \bigcdot\sub{\rr}) \coloneq \sum_i   \Phi^{(i)} (\bigcdot\sub{\h}) \otimes P_i \bigcdot\sub{\rr} P_i \, ,
    \end{align*}
where $\{P_i\}$ are rank-1 orthocomplete projections on $\rr$, is a rank non-decreasing CP map acting in $\h \otimes \rr$.     
\unskip\nobreak\hfill $\square$ \end{lemma}
\begin{proof}
Let $\{C_1^{(i)}\}$ and $\{C_2^{(i)}\}$ be operators on $\h$, and define  
   \begin{align*}
       &D_1 = \sum_i  C_1^{(i)} \otimes P_i \, , &D_2 = \sum_i C_2^{(i)} \otimes P_i \, .
   \end{align*}
We observe that 
\begin{align*}
    \Lambda_{D_1, D_2} (\one\sub{\h} \otimes \one\sub{\rr}) &= \sum_i \Phi^{(i)}_{C_1^{(i)}, C_2^{(i)}}(\one\sub{\h}) \otimes  P_i \, , \\
    [\Lambda_{D_1, D_2}]^* (\one\sub{\h} \otimes \one\sub{\rr}) &= \sum_i  [\Phi^{(i)}_{C_1^{(i)}, C_2^{(i)}}]^*(\one\sub{\h}) \otimes P_i \, .
\end{align*}
Now note that $(\sum_i A^{(i)} \otimes P_i - \one\sub{\h} \otimes \one\sub{\rr})^2 = \sum_i   (A^{(i)} - \one\sub{\h})^2 \otimes P_i$. As such, 
\begin{align*}
    DS(\Lambda_{D_1,D_2}) = \sum_i DS\left(\Phi^{(i)}_{C_1^{(i)}, C_2^{(i)}}\right). 
\end{align*}
Since $\{\Phi^{(i)}\}$ are rank non-decreasing, by \lemref{lemma:iff-rank-non-dec} for every $\epsilon >0$ we may choose the operators $\{C_1^{(i)}\}$ and $\{C_2^{(i)}\}$  so that  $DS(\Phi^{(i)}_{C_1^{(i)}, C_2^{(i)}}) \leqslant \epsilon^2/ \dim(\rr)$ for all $i$. In such a case, we have that $DS(\Lambda_{D_1,D_2}) \leqslant \epsilon^2$. As such, $\Lambda$ is rank non-decreasing. 
\end{proof}

%-------------------------------------------------------------
%-----Appendix: convexity-------------------------------------
%-------------------------------------------------------------

\section{Geometric properties}\label{app:convexity}

Here, we wish to show that for any thermodynamic constraint  $C \in \{\I, \II, \III\}$ as per  \defref{defn:thermo-consistent-operation}, the set of operations $\opset_C(\hs)$ is convex. 

\

\begin{prop}
    Consider the set of  operations, $\opset_C(\hs)$,  for the thermodynamic constraint $C \in \{\I, \II, \III\}$ as per  \defref{defn:thermo-consistent-operation}. These sets are convex. That is, for any  pair of operations $\Phi_i \in \opset_C(\hs)$, $i=1,2$, and any $0 \leqslant \lambda \leqslant 1$, there exists a process obeying the constraint $C$ which realises the operation $\Phi(\bigcdot) \coloneq \lambda \,  \Phi_1 (\bigcdot) + (1-\lambda) \Phi_2(\bigcdot)$.
\unskip\nobreak\hfill $\square$ \end{prop}
\begin{proof}
 Note that the case of  $\lambda =0 $ or $\lambda = 1$ is trivial, so we shall consider only $0 < \lambda < 1$. Let $(\h\sub{\aa_i}, \xi_i, \ee_i, Z_i)$, $i = 1,2$,  be a process, under constraint $C$, that   implements $\Phi_i$ as in \eq{eq:operation-implementation}. Now consider the process $(\ha, \xi, \ee, Z)$. Let us choose $\ha = \h\sub{\aa_1} \otimes \h\sub{\aa_2} \otimes \rr$,   with $\rr = \co^2$, which has an orthonormal basis $\{|1\>, |2\>\}$. Denote $P_i \equiv \proj{i}$. Prepare this system in state $\sigma = \lambda P_1 + (1- \lambda) P_2$, which is strictly positive, so that   $\xi = \xi_1 \otimes \xi_2 \otimes \sigma$ is strictly positive.

Choose the channel $\ee$ acting in $\hs  \otimes \h\sub{\aa_1} \otimes \h\sub{\aa_2} \otimes \rr $ as
\begin{align*}
    \ee( \bigcdot\sub{\s+ \aa_1 + \aa_2 } \otimes \bigcdot\sub{\rr}  ) =    \ee_1 \otimes \idch\sub{\aa_2} (\bigcdot\sub{\s+ \aa_1 + \aa_2}) \otimes P_1 \bigcdot\sub{\rr} P_1 + \ee_2 \otimes \idch\sub{\aa_1} (\bigcdot\sub{\s+ \aa_1 + \aa_2}) \otimes P_2 \bigcdot\sub{\rr} P_2 \, . 
\end{align*}
Note that 
\begin{align*}
    \ee(\onesys \otimes \one\sub{\aa_1} \otimes \one\sub{\aa_2} \otimes \one\sub{\rr}) = \ee_1(\onesys \otimes \one\sub{\aa_1}) \otimes \one\sub{\aa_2} \otimes P_1 +   \ee_2(\onesys \otimes \one\sub{\aa_2}) \otimes \one\sub{\aa_1} \otimes P_2\, .
\end{align*}
It is simple to verify that if $\ee_i$ are bistochastic, then $\ee$ is bistochastic, whereas if $\ee_i$ are strictly positive, then $\ee$ is strictly positive. Finally,  if $\ee_i$ are rank non-decreasing, then by \propref{prop:complete-rank-non-dec} and  \lemref{lemma:conditional-projection-rank-non-dec}, $\ee$ is rank non-decreasing. It follows that $(\ha, \xi, \ee, Z)$ is subject to the thermodynamic constraint $C$. 

Now choose the effect $Z =  Z_1 \otimes \one\sub{\aa_2} \otimes P_1 +  \one\sub{\aa_1} \otimes Z_2 \otimes P_2$. The process therefore  implements an operation $\Phi$ through
\begin{align*}
  \Phi(\bigcdot) &\coloneq   \tra [(\onesys \otimes Z)\ \ee(\,\bigcdot\, \otimes \xi)]  \\
    & =  \lambda \tra [(\onesys \otimes Z)\  \ee_1(\,\bigcdot\, \otimes  \xi_1) \otimes \xi_2 \otimes P_1] + (1- \lambda) \tra [(\onesys \otimes Z)\  \ee_2 (\,\bigcdot\, \otimes  \xi_2 ) \otimes \xi_1 \otimes P_2] \\
    & =   \lambda \tr\sub{\aa_1} [(\onesys \otimes Z_1)\ \ee_1 (  \,\bigcdot\, \otimes \xi_1)] + (1-\lambda )\tr\sub{\aa_2} [(\onesys \otimes Z_2)\ \ee_2 (  \,\bigcdot\, \otimes \xi_2)] \\
    & = \lambda \, \Phi_1(\,\bigcdot\,) + (1-\lambda) \Phi_2(\,\bigcdot\,) \, .
\end{align*}
Therefore, $\opset_C(\hs)$ is convex which concludes the proof.
\end{proof}

Now, we shall show that for any $C \in \{\I, \II, \III\}$, $\opset_C(\hs)$ is closed under composition.

\

\begin{prop}
Consider the set of  operations $\opset_C(\hs)$ for  the thermodynamic constraint $C \in \{\I, \II, \III\}$ as per  \defref{defn:thermo-consistent-operation}. These sets are closed under composition. That is, for any pair of  operations $\Phi_i \in \opset_C(\hs)$, $i = 1,2$,  there exists a process $(\ha, \xi, \ee, Z)$ obeying the constraint $C$ which realises the operation $\Phi(\,\bigcdot\,) \coloneq \Phi_2 \circ \Phi_1 (\,\bigcdot\,)$.    
\unskip\nobreak\hfill $\square$ \end{prop}

\begin{proof}

Let $(\h\sub{\aa_i}, \xi_i, \ee_i, Z_i)$, $i = 1,2$,  be a process, under constraint $C$, that   implements $\Phi_i$ as in \eq{eq:operation-implementation}. Now consider the process $(\ha, \xi, \ee, Z)$. Choose $\ha = \h\sub{\aa_1} \otimes \h\sub{\aa_2}$, $\xi = \xi_1 \otimes \xi_2$, and $\ee = (\idch\sub{\aa_1} \otimes\, \ee_2 )\circ(\ee_1 \otimes \idch\sub{\aa_2})$. Clearly, $\xi$ is strictly positive. On the other hand, if $\ee_i$ are strictly positive, rank non-decreasing, or bistochastic, then so too is $\ee$. The process $(\ha, \xi, \ee, Z)$ is therefore consistent with constraint $C$. Finally, choosing  $Z = Z_1 \otimes Z_2$, we get
\begin{align*}
\Phi(\,\bigcdot\,) & = \tra[ (\onesys \otimes Z)\ \ee(\,\bigcdot\, \otimes Z)] \\
  &= \tr\sub{\aa_1+\aa_2}\left[\onesys \otimes Z_1 \otimes Z_2\ \big(\idch\sub{\aa_1} \otimes \ee_2 \big)\circ \big(\ee_1 \otimes \idch\sub{\aa_2} \big) \big(\,\bigcdot\, \otimes\, \xi_1 \otimes \xi_2 \big) \right] \\
    &= \tr\sub{\aa_1+\aa_2}\left[\onesys \otimes Z_1 \otimes Z_2\ \big(\idch\sub{\aa_1} \otimes \ee_2 \big) \big(\ee_1 \big(\,\bigcdot\, \otimes\, \xi_1  \big) 
 \otimes \xi_2 \big) \right] \\ 
    & = 
    \tr\sub{\aa_2} \left[\onesys \otimes  Z_2 \  \ee_2 \big(\tr\sub{\aa_1}[\onesys \otimes Z_1\ \ee_1 (\,\bigcdot\, \otimes \xi_1)] \otimes \xi_2 \big) \right] \\
    & = \Phi_2 \circ \Phi_1 (\,\bigcdot\,)\, .
\end{align*}
\end{proof}

%-------------------------------------------------------------
%-----Appendix: restriction-map-------------------------------
%-------------------------------------------------------------

\section{Restriction maps, and conjugate channels }
\label{app:restriction-map}

In this section, we shall introduce some concepts and notation that will be frequently employed in the subsequent proofs. Let us introduce the unital CP map $\Gamma^\ee_\xi : \lo(\hs \otimes \ha) \to \lo(\hs)$, defined as

\begin{align}\label{eq:gamma-E}
\Gamma^\ee_\xi := \Gamma_\xi \circ \ee^*\, .    
\end{align}
Here, $\Gamma_\xi : \lo(\hs \otimes \ha) \to \lo(\hs)$ is a unital CP map, referred to as a conditional expectation, or restriction map, with respect to $\xi$, which reads
\begin{align*}
    \Gamma_\xi (\bigcdot) \coloneq \tra[(\onesys \otimes \xi) \,   \,\bigcdot\,]\, . 
\end{align*}

Using \eq{eq:gamma-E}, we may write the dual  operations of the instrument $\ii$ implemented by $(\ha, \xi, \ee, \Z)$, defined in \eq{eq:instrument-implementation},    as well as the dual  of the operations $\Phi$ implemented by $(\ha, \xi, \ee, Z)$, defined in \eq{eq:operation-implementation}, as
\begin{align}
&\ii_x^*(\bigcdot) =  \Gamma_\xi^\ee(\,\bigcdot\, \otimes Z_x) \, ,  &\Phi^*(\bigcdot) = \Gamma_\xi^\ee(\,\bigcdot\, \otimes Z)\,  .
\end{align}

\

Now let us introduce the channel  $\Lambda: \lo(\hs) \to \lo(\ha)$, and its dual $\Lambda^*: \lo(\ha) \to \lo(\hs)$, defined as 
\begin{align}\label{eq:conjugate-channel}
&\Lambda(\bigcdot) \coloneq \trs[\ee(\,\bigcdot\, \otimes \xi)] \, ,  
&\Lambda^*(\bigcdot) \coloneq \Gamma_\xi^\ee(\onesys \otimes\, \bigcdot\,) \, .
\end{align}
$\Lambda$ is referred to as the conjugate channel (also called complementary channel \cite{Holevo2007})  to the $\E$-channel $\ii_\xx(\bigcdot) \equiv  \tra[\ee(\,\bigcdot\, \otimes \xi)]$.  $\Lambda(\rho)$ is the state of the auxiliary system after it has interacted with the system, when the system is initially prepared in the state $\rho$. It is easily verified that for an operation $\Phi(\bigcdot) = \tra[(\onesys \otimes Z)\ \ee(\,\bigcdot\, \otimes \xi)]$ compatible with the effect $E$, it holds that $\tr[E \rho] = \tr[\Phi(\rho)] = \tr[Z \ \Lambda(\rho)]$. Indeed,  $E = \Lambda^*(Z)$. 

\

\begin{lemma}\label{lemma:Gamma-E-positive}
Let  $\xi$ be a strictly positive state on $\ha$ and $\ee$ be a strictly positive channel acting in $\hs \otimes \ha$.  The following hold: 

\begin{enumerate}[(i)]
    \item For all $B \in \lo(\hs \otimes \ha)$, it holds that $\Gamma^\ee_\xi(B^*B) = \zero \iff B = \zero$ where $\Gamma^\ee_\xi$ is defined in \eq{eq:gamma-E}.

    \item    $\Lambda$ defined in \eq{eq:conjugate-channel} is strictly positive, and so  $E = \zero \iff Z = \zero$.
\end{enumerate}
\unskip\nobreak\hfill $\square$ \end{lemma}
\begin{proof}
$(i)$:  Note that $\Gamma_\xi^\ee \coloneq \Gamma_\xi \circ \ee^*$ is dual to the channel $\ee \circ \Upsilon_\xi$, where $\Upsilon_\xi : \lo(\hs) \to \lo(\hs \otimes \ha), \rho \mapsto \rho \otimes \xi$.  $\Upsilon_\xi$ is a strictly positive channel for   any strictly positive state $\xi$. It follows that the composition $\ee \circ \Upsilon_\xi$ is a strictly positive channel. The statement follows from item (i) of  \lemref{lemma:equivalent-dual-positive-rank-non-dec}.

$(ii)$: By item $(i)$  it holds that $\Lambda^*(A^*A) = \Gamma^\ee_\xi(\onesys \otimes A^*A) = \zero \iff A = \zero$, and so by \lemref{lemma:equivalent-dual-positive-rank-non-dec} it follows that $\Lambda$ is strictly positive. That $E = \zero \iff Z = \zero$ follows trivially from the fact that $E = \Lambda^*(Z)$. 

\end{proof}

%-------------------------------------------------------------
%-----Appendix: weak-third-law--------------------------------
%-------------------------------------------------------------

\section{The weak third law}\label{app:weak-third-law}

In this section, we  obtain necessary and sufficient conditions for an  operation to be consistent with the weak third law, i.e., operations that admit a process $(\ha, \xi, \ee, Z)$  as per \eq{eq:operation-implementation}, where $\xi$ is a strictly positive state on $\ha$ and $\ee$ is a strictly positive channel acting in $\hs \otimes \ha$. That is, we shall characterise $\opset_\I( \hs)$ as per \defref{defn:thermo-consistent-operation}.

\

\begin{lemma}\label{lemma:measurement-weak-third-law}
Let $E \ne \zero$ be an effect on $\hs$, and  $\Phi$ be an $E$-compatible operation in $\opset_\I(\hs)$, as per \defref{defn:thermo-consistent-operation}. Then 
\begin{enumerate}[(i)]

\item 

$\Phi$ is a strictly positive operation.

\item 

Let  $P$ be the  projection on the support  of $E$. For every  state $\rho$ on $\hs$ such that $P \rho P$  has full rank in $P \hs$, $\Phi(\rho)$ has  full rank in $\hs$. 
    \end{enumerate}
\unskip\nobreak\hfill $\square$ \end{lemma}
\begin{proof}
 \begin{enumerate}[(i)]
     \item

Since $E \ne \zero$ it holds that $Z \ne \zero$.  Now note that $\Phi^*(A^*A) = \Gamma_\xi^\ee(A^*A \otimes Z)$, where $\Gamma^\ee_\xi$ is defined in \eq{eq:gamma-E}. It follows from \lemref{lemma:Gamma-E-positive} that $\Phi^*(A^*A) = \zero$ if and only if $A^*A \otimes Z  = \zero$, which holds if and only if $A = \zero$. The statement follows from \lemref{lemma:equivalent-dual-positive-rank-non-dec}.

\item 

We may always write $\Phi^*(\bigcdot) = \sqrt{E } \ \Xi^*(\bigcdot) \sqrt{E }$ for some channel $\Xi$ acting in $\hs$. It follows that $\Phi^*(\bigcdot) = P \Phi^*(\bigcdot)P $, and so for any $A \in \lo(\hs)$, $\Phi^*(A^*A) = P \Phi^*(A^* A) P \in \lo(P \hs)$.  
   By item (i),  given   any $\rho$ for which $P \rho P $ has full-rank in $P \hs$, it follows that $\tr[\Phi^*(A^* A) \rho] = 0 \iff A = \zero$.  By writing  $\tr[A^* A \, \Phi(\rho)] = \tr[\Phi^*(A^* A) \rho]$, it follows that $\tr[A^* A \Phi(\rho)] =0 \iff A = \zero$,  and so $\Phi(\rho)$ must have  full rank in $\hs$. 
 \end{enumerate}   
\end{proof}

Note that item (ii) implies that if $\rank{E} = 1$, then for any state $\rho$ such that $\tr[E \rho] >0$, including a pure state, it will hold that $\rank{\Phi(\rho)} = \dim(\hs)$.

\

\begin{prop}\label{prop:necessary-sufficient-weak-third-law}
    An  operation $\Phi$, compatible with $E \ne \zero$,  exists in $\opset_\I(\hs)$ if and only if $\Phi$ is strictly positive. Similarly, an $\E$-compatible instrument $\ii \coloneq \{\ii_x : x \in \xx\}$ such that  $E_x \ne \zero$ for all $x$ exists in $\inset_\I(\hs)$ if and only if  $\ii_x$ is strictly positive for all $x$.
\unskip\nobreak\hfill $\square$ \end{prop}
\begin{proof}
   The only if statement for both operations and instruments follows from  \lemref{lemma:measurement-weak-third-law}, while the if statement for channels follows trivially by noting that the interaction channel $\ee = \Phi \otimes \idchapp$ is strictly positive if $\Phi$ is, and that by choosing $Z = \oneapp$  such an interaction implements $\Phi$. So we shall now show  the if statement for operations and instruments. Consider the observable $\E \coloneq \{E_x : x \in \xx \}$ for $\xx = \{1,\cdots,N\}$, and the $\E$-compatible instrument $\ii$. Let us identify $E_1$ and $\ii_1$ with the particular  effect $E$ and its $E$-compatible operation $\Phi$, respectively. Since $E_x \ne \zero$, then $\ii_x$ can always be chosen to be strictly positive, see \corref{app-cor:strict-positive-for-any-effect}.   Choose $\ha$ with $\dim(\ha) = |\xx| = N$, and let $\{|x\> : x = 1, \dots , N\}$ be an orthonormal basis for $\ha$.  Let $\ee$ be defined as
    \begin{align*}
        \ee(A \otimes B) = \sum_{x=1}^{N} \ii_x(A) \otimes \tr[B] \proj x
    \end{align*}
    for all $A \in \lo(\hs), B \in \lo(\ha)$. It is readily verified that $\ee$ is a channel, and that if we choose $Z_x=\proj x$, then 
    \begin{align*}
       \ii_x(\bigcdot) =   \tra[\onesys \otimes \proj x\ \ee(\,\bigcdot\, \otimes \xi)] 
    \end{align*}
for any state $\xi$.  All that remains to be shown is that $\ee$ is strictly positive. This is guaranteed to be the case if $\ee(\onesys \otimes \oneapp) > \zero$. But it holds that 
\begin{align*}
  \ee(\onesys \otimes \oneapp) = N  \sum_x \ii_x(\onesys) \otimes \proj x \, .   
\end{align*}
 Since $\ii_x(\onesys) > \zero$  for all $x$, and $\{|x\>\}$ spans $\ha$, it holds that $\ee(\onesys\otimes\oneapp)>\zero$ which completes the proof.
\end{proof}

%-------------------------------------------------------------
%-----Appendix: strong-third-law------------------------------
%-------------------------------------------------------------
\section{The strong third law}\label{app:strong-third-law}

In this section, we  obtain necessary and sufficient conditions for an  operation to be consistent with the strong third law, i.e., operations that admit a process $(\ha, \xi, \ee, Z)$ as per \eq{eq:operation-implementation}, where $\xi$ is a strictly positive state on $\ha$ and $\ee$ is a rank non-decreasing  channel acting in $\hs \otimes \ha$. That is, we shall characterise $\opset_\II( \hs)$ as per \defref{defn:thermo-consistent-operation}.  Let us first introduce the following useful result, which follows from the  weak subadditivity of the R\'enyi-zero entropy as given in Lemma 4.3 of Ref. \cite{VanDam2002}:

\

\begin{lemma}\label{lemma:rank-partial-trace}
 For all positive semi-definite operators $\rho$ on $\hs\otimes\ha$, the following holds:
 \begin{align*}
     \rank{\rho} \leqslant \rank{\tra[\rho]} \rank{\trs[\rho]} \, .
 \end{align*}
\unskip\nobreak\hfill $\square$ \end{lemma}

Using the above, we are able to obtain the following:

\

\begin{lemma}\label{lemma:measurement-strong-third-law}
Let  $\Phi$ be an $E$-compatible operation in $\opset_\II(\hs)$, as per \defref{defn:thermo-consistent-operation}. Then 
 
\begin{enumerate} [(i)]

\item 

If $Z > \zero$, then $E > \zero$, and  $\Phi$ is a rank non-decreasing operation.

\item 

Let $E$ be a non-trivial norm-1 effect. Then for every state $\rho$ on $\hs$ it holds that 
\begin{align*}
  \tr[E \rho] = 1 \implies  \rank{\Phi(\rho)}  > \rank{\rho}.
\end{align*}

\end{enumerate} 
\unskip\nobreak\hfill $\square$ \end{lemma}
\begin{proof}

\begin{enumerate}[(i)]

\item 

Let us first note that if $Z  > \zero$, then $\tr[E  \rho] = \tr[Z  \Lambda(\rho)] >0$ for all $\rho$, where $\Lambda$ is defined in \eq{eq:conjugate-channel}. It follows that $E  > \zero$, and
\begin{align*}
    \rank{\rho}\dim(\ha) & = \rank{\rho \otimes \xi} \\
    & \leqslant \rank{ \ee( \rho \otimes \xi)} \\
    & = \rank{\onesys \otimes  \sqrt{Z} \,  \ee( \rho \otimes \xi) \,  \onesys \otimes  \sqrt{Z}} \\
    & \leqslant \rank{\Phi(\rho)} \rank{\sqrt{Z} \Lambda(\rho) \sqrt{Z}} \\
    & \leqslant  \rank{\Phi(\rho)} \dim(\ha)
\end{align*}
for all $\rho$, and so $\rank{\rho} \leqslant \rank{\Phi(\rho)}$. Here, in the third line we have used the fact that if $Z > \zero $ then $\rank{\sqrt{Z} \sigma \sqrt{Z}} = \rank{\sigma}$, and in the fourth line we use \eq{eq:operation-implementation} and   \lemref{lemma:rank-partial-trace}.

\item

If $Z = \oneapp$, then $\Phi$ is a channel, which is compatible with a trivial effect $E = \onesys$. Therefore, $Z \ne  \oneapp$. Since $E$ is norm-1 and $\Lambda^*$    defined in \eq{eq:conjugate-channel} is completely positive and unital, then  $1 = \|E \| = \|\Lambda^*(Z )\| \leqslant \| Z \| \leqslant 1$, and so $\|Z  \| = 1$. It follows that $Z$ is a non-trivial norm-1 effect, and the eigenvalue-1 eigenspace of $Z$ is strictly smaller than $\ha$. 

Let $\rho$ be a state which has support only in the eigenvalue-1 eigenspace of $E $, so that $\tr[\Phi(\rho)] = \tr[E  \rho] =1$. By the probability reproducibility condition, it follows that
\begin{align*}
\tr[(\onesys \otimes Z)\ \ee(\rho \otimes \xi)] =     \tr[Z  \Lambda(\rho)] = \tr[\Phi(\rho)] = 1 \, .
\end{align*}
 Since $Z$ is an effect and $\Lambda(\rho)$ is a state, then $\tr[Z  \Lambda(\rho)] = 1$ implies that $\Lambda(\rho) = Z  \Lambda(\rho)$ must  have support only in the eigenvalue-1 eigenspace of $Z $, and so $\rank{\Lambda(\rho)} < \dim(\ha)$. Indeed, $ \ee(\rho \otimes \xi)$ also has support only in the eigenvalue-1 eigenspace of $\onesys \otimes Z$ so that $(\onesys \otimes Z)\, \ee(\rho \otimes \xi) =  \ee(\rho \otimes \xi)$.   It follows that 
\begin{align*}
    \rank{\rho} \dim(\ha) & \leqslant \rank{\ee(\rho \otimes \xi)} =  \rank{(\onesys \otimes Z)\ \ee(\rho \otimes \xi)} \leqslant \rank{\Phi(\rho)} \rank{\Lambda(\rho)} \, ,
\end{align*}
where the final inequality follows from \eq{eq:operation-implementation} and \lemref{lemma:rank-partial-trace}. Therefore, 
\begin{align*}
 \frac{\rank{\Phi(\rho)} }{\rank{\rho}} \geqslant \frac{\dim(\ha)}{\rank{\Lambda(\rho)}} >1 \,.  
\end{align*}

\end{enumerate}
\end{proof}

\begin{corollary}
    A channel $\Phi$ exists in $\opset_\II(\hs)$ if and only if it is rank non-decreasing. 
\unskip\nobreak\hfill $\square$ \end{corollary}
\begin{proof}
The if statement follows trivially by observing that the interaction channel  $\ee = \Phi \otimes \idchapp$ is rank non-decreasing if $\Phi$ is (see \propref{prop:complete-rank-non-dec}), so that by choosing $Z = \oneapp$ the process implements the channel $\Phi$. The only if statement follows from item (i) of \lemref{lemma:measurement-strong-third-law}, together with the fact that  if $Z = \oneapp$ (and hence $Z > \zero$) then $\Phi$ is a channel, and since $\ee$ is rank non-decreasing, $\Phi$ cannot be a channel if $Z \ne \oneapp$. To see the second claim, note that $\tr[\Phi(\rho)] = 1$ if and only if $\tr[(\onesys \otimes Z)\ \ee(\rho \otimes \xi)] = 1$ which,  when $\rho > \zero$ implying that  $\ee(\rho \otimes \xi) > \zero$, will be satisfied if and only  if $Z = \oneapp$.  
\end{proof}

While item (i) of \lemref{lemma:measurement-strong-third-law} shows that a channel is consistent with the strong third law if and only if it is rank non-decreasing, not all operations consistent with the strong third law are rank non-decreasing, and not all rank non-decreasing operations are consistent with the strong third law. 

\

\begin{prop}\label{prop:go-no-go-FP-op}
 \begin{enumerate}[(i)]
     \item There exist operations $\Phi \in \opset_\II(\hs)$ that are not rank non-decreasing.

     \item Let $\Phi \in \opset_\II(\hs)$ be an $E$-compatible operation.      If $E \ne \onesys$, then $\ff(\Phi) = \ff(\Phi^*) = \zero$.

     \item There are rank non-decreasing operations that exist in $\opset_\I(\hs)$ but not in $\opset_\II(\hs)$.
 \end{enumerate}   
\unskip\nobreak\hfill $\square$ \end{prop}
\begin{proof}
\textit{(i)}:  An operation $\Phi$ is rank non-decreasing only if it is compatible with a strictly positive effect $E$, see \corref{app-cor:strict-positive-for-any-effect}. But there exist operations in $\opset_\II(\hs)$ that are compatible with effects $E$ that are not strictly positive, and hence are not rank non-decreasing. Consider the process $(\ha, \xi, \ee, Z)$ where $\ha = \hs$ and $\ee$ is a unitary (and hence rank non-decreasing) swap channel. Then for any effect $Z$, the process implements the operation $\Phi(\bigcdot) = \tr[Z \bigcdot] \xi$, which is compatible with the effect $E = Z$. Even though $\xi$ is strictly positive, unless $Z > \zero$ then $\Phi$ is not rank non-decreasing.

\

\textit{(ii)}:    
By \lemref{lemma:null-fixed-point}, $\ff(\Phi)$ contains a state $\rho$ only if $\| E \| = 1$. 
 Now assume that $\|E\|=1$ but $E \ne \onesys$.  Since $\Phi(\rho) = \rho \implies  \tr[E \rho] = \tr[\Phi(\rho)] = \tr[\rho] = 1$, then  $\rho$ is a fixed state of $\Phi$ only if  $\tr[E \rho] = 1$.  But  by item (ii) of \lemref{lemma:measurement-strong-third-law} it holds that $\tr[E \rho ] = 1 \implies \rank{\Phi(\rho)} > \rank{\rho}$, and so $\tr[E \rho] = 1 \implies \Phi(\rho) \ne \rho$.  But as shown in Theorem 6.5 of Ref. \cite{Wolf2012}, if there exists any $ \lo(\hs) \ni A \ne \zero$ such  $\Phi(A) = A$, then there exists a state $\rho$ such that $\Phi(\rho) = \rho$. It follows that $\ff(\Phi) = \zero$.  Finally, since $\dim(\ff(\Phi^*)) = \dim(\ff(\Phi))$, then $\ff(\Phi^*) = \zero$. 

 \

\textit{(iii)}: 
If   $E > \zero$,  then the L\"uders operation $\Phi^L(\bigcdot) \coloneq \sqrt{E} \bigcdot \sqrt{E}$ is rank non-decreasing (in fact it preserves the rank) and hence strictly positive, and thus by \propref{prop:necessary-sufficient-weak-third-law} it exists in $\opset_\I(\hs)$. This is so even if  $\| E\| =1$ and $E \ne \onesys$; but for such an effect, by \lemref{lemma:effect-support-probability-1}, for any state $\rho$ such that $\tr[E \rho] = 1$, it will hold that $\sqrt{E} \rho \sqrt{E} = \rho$. By item (ii), this operation does not exist in $\opset_\II(\hs)$.

\end{proof}

\begin{remark}
    Compare item (ii) of the above with item (ii) of \lemref{lemma:null-fixed-point}. The fact that $\tr[E \rho ] = 1 \implies \rank{\Phi(\rho)} > \rank{\rho}$ implies that for any projection $P$ satisfying $ P = EP$, the operation $\Phi_P(\bigcdot) \coloneq P \Phi(P \bigcdot P) P$ does not equal  $\Phi(P \bigcdot P)$. 
\unskip\nobreak\hfill $\square$ \end{remark}

\

While not all rank non-decreasing operations are consistent with the strong third law, the following shows  that any rank non-decreasing operation compatible with an indefinite effect   is. Note that such operations do not have any non-vanishing fixed points, owing to \lemref{lemma:null-fixed-point}.

\

\begin{prop}\label{prop:existence-rank-non-dec-strong-third-law}
Any rank non-decreasing operation $\Phi$ compatible with an indefinite effect $\zero < E < \onesys$ exists in $\opset_\II(\hs)$. Similarly, any instrument $\ii\coloneq \{\ii_x : x \in \xx\}$ such that $\ii_x$ are rank non-decreasing operations compatible with indefinite effects for all $x$ exists in $\inset_\II(\hs)$.  
\unskip\nobreak\hfill $\square$ \end{prop}
\begin{proof}
Consider again the process introduced in   \propref{prop:necessary-sufficient-weak-third-law}. Let $\E \coloneq \{E_x : x\in \xx\}$, with $\xx = \{1, \dots, N\}$,  be an indefinite observable, with $\ii$ an $\E$-compatible instrument. Let $E_1$ and $\ii_1$ be identified with the particular  indefinite effect $E$ and its  $E$-compatible operation $\Phi$, respectively. Since all effects $E_x$ are indefinite, i.e., $\zero < E_x < \onesys$, then all operations $\ii_x$ can be chosen to be rank non-decreasing. See \corref{app-cor:strict-positive-for-any-effect}.    Choose $\ha$ with $\dim(\ha) = |\xx| = N$, and let $\{|x\> : x = 1, \dots , N\}$ be an orthonormal basis for $\ha$.  Let $\ee$ be defined as
    \begin{align*}
        \ee(A \otimes B) = \sum_{x=1}^{N} \ii_x(A) \otimes \tr[B] \proj x
    \end{align*}
    for all $A \in \lo(\hs), B \in \lo(\ha)$. As before, $\ee$ is a channel, and choosing  $Z_x=\proj x$ implements $\ii_x$. All that is left to show is that $\ee$ is rank non-decreasing.

 Let $\rho\sub{\s \aa}$ denote any state in $\s(\hs \otimes \ha)$, with $\rho\sub{\s} := \tra[\rho\sub{\s\aa}]$ and $\rho\sub{\aa} := \trs[\rho\sub{\s\aa}]$ its reduced states. It holds that 
\begin{align*}
\ee(\rho\sub{\s\aa}) = \sum_{x=1}^{N} \ii_x(\rho\sub{\s}) \otimes \proj{x}  \, .  
\end{align*}
But since $\proj x$ are mutually orthogonal rank-1 projections, and $\ii_x$ are rank non-decreasing, we have that 
\begin{align*}
    \rank{\ee(\rho\sub{\s\aa}) } & = \sum_{x=1}^{N} \rank{\ii_x(\rho\sub{\s})} \geqslant \sum_{x=1}^N \rank{\rho\sub{\s}}  = \rank{\rho\sub{\s}} \dim(\ha)  \geqslant \rank{\rho\sub{\s}}  \rank{\rho\sub{\aa}}  \geqslant \rank{\rho\sub{\s\aa}}
\end{align*}
for all $\rho\sub{\s\aa}$, where the final inequality follows from \lemref{lemma:rank-partial-trace}. Therefore,  $\ee$ is rank non-decreasing. Note that by \lemref{lemma:measurement-strong-third-law}, this implies that   the $\E$-channel   $\ii_\xx(\bigcdot) \coloneq \sum_x \ii_x (\bigcdot) \equiv \tra[\ee(\bigcdot \otimes \xi)]$ is  a rank non-decreasing channel.
\end{proof}

%-------------------------------------------------------------
%-----Appendix: strong-adiabatic------------------------------
%-------------------------------------------------------------

\section{Full consistency with thermodynamics}\label{app:strong-adbatic}

A subset of channels that are guaranteed to have a strictly positive fixed state are bistochastic ones, which preserve the complete mixture. It is clear that all bistochastic channels exist in $\opset_\III(\hs)$. This follows from the fact that if $\Phi$ is a bistochastic channel, then the interaction channel  $\ee = \Phi \otimes \idchapp$ is also bistochastic, and that for any strictly positive state preparation $\xi$ it holds that $\Phi(\bigcdot) = \tra[\ee(\bigcdot \otimes \xi)]$, and so the process $(\ha, \xi, \ee, Z = \oneapp)$  will implement the channel $\Phi$. Surprisingly, as we shall soon see,  every channel $\Phi \in \opset_\III(\hs)$ is guaranteed to have a strictly positive state, even if not bistochastic. In order to show this, we first need to introduce some basic concepts regarding the classical action of channels. 

\

\begin{definition}[Classical action]
\label{def:classical-action}
    Let $\Phi$ be a channel acting in $\h$, and let $\varphi \coloneq \{|\varphi_m\>\}$ be an orthonormal basis  that spans $\h$. The $\varphi$-classical action of $\Phi$  is defined as the   matrix $\cT \equiv [\cT _{m,n}]$ with  elements
    \begin{align}
    \label{eq:classical-action}
        \cT_{m,n} \coloneq \<\varphi_m| \Phi(\proj{\varphi_n}) |\varphi_m\> \in [0,1]\, .
    \end{align}    
    \unskip\nobreak\hfill $\square$
\end{definition}
Since  $\Phi$ is trace-preserving, then $\sum_m \cT_{m,n} = 1$. That is, the classical action of a channel is a (column) stochastic matrix. If $\Phi$ is a bistochastic channel, then it also holds that $\sum_n \cT_{m,n} = 1$, and so any $\varphi$-classical action of a bistochastic channel is a bistochastic (or doubly stochastic) matrix. It is straightforward to show that the $\varphi$-classical action $\cT$ can be obtained by
\begin{equation}
\label{eq:T-kraus}
    \cT=\sum_a K_a\odot \overline{K}_a\, ,
\end{equation}
where $\{K_a\}$ is any  Kraus representation of the channel $\Phi$ written as matrices in the $\varphi$ basis, $\overline{A}$ is the $\varphi$-basis matrix representation of an operator $A \in \lo(\h)$ with its elements complex-conjugated,    and $\odot$ denotes the Hadamard (entry-wise) product of two matrices.

\

\begin{definition}[reducible stochastic matrix]
A $d \times d$ stochastic matrix $\cT$ is reducible if and only if there exists a permutation matrix $\Pi$ such that 
\begin{equation}
\label{eq:block_form_matrix}
   \Pi \cT\Pi^{-1}=\left(\begin{array}{@{}c|c@{}}
  \begin{matrix}
     A
  \end{matrix}
  & B \\
\hline
   \zero &
  \begin{matrix}
   C
  \end{matrix}
\end{array}\right)\, ,
\end{equation} 
where $A$ and $C$ are $d_A \times d_A$ and $d_C \times d_C$ square matrices, respectively. Otherwise, $\cT$ is irreducible.
\unskip\nobreak\hfill $\square$
\end{definition}

Note that in the above,  $A$ is itself a stochastic matrix. Moreover, if $\cT$ is the $\varphi$-classical action of a channel $\Phi$, then $\Pi$ can be interpreted as a relabelling of the elements of $\varphi$ so that $\cT$ admits the block structure on the right hand side of \eq{eq:block_form_matrix}.

\

\begin{lemma}
\label{lem:not-bistochasticity}
    Let $\cS$ be a $d \times d$ bistochastic matrix, and $\Pi$ a permutation.  Assume that 
\begin{align*}%\label{eq:bistochastic-irreducible}
    \Pi \cS\Pi^{-1}=\left(\begin{array}{@{}c|c@{}}
  \begin{matrix}
     A
  \end{matrix}
  & B \\
\hline
   \zero &
  \begin{matrix}
   C
  \end{matrix}
\end{array}\right), 
\end{align*}    
where $A$ and $C$ are $d_A \times d_A$ and $d_C \times d_C$ square matrices, respectively. Then $B = \zero$, while  $A$ and $C$ are bistochastic \cite{Perfect1965}.
\unskip\nobreak\hfill $\square$ \end{lemma}

\

The definition of reducibility of stochastic matrices, and an inductive application of the above argument, implies that every bistochastic matrix $\cS$ admits a permutation $\Pi$ such that $\Pi \cS \Pi^{-1}=\oplus_\beta \cS_\beta$, where $\cS_\beta$ are irreducible bistochastic matrices. Note that if $\cS$ is irreducible, then the index set $\{\beta\}$ is a singleton, so that $\cS_\beta = \cS$. 

\

\begin{lemma}\label{lemma:bipartite-bistochastic-reducibility}
    Let $\cS$ be a bipartite bistochastic matrix on $\re^{d_\s} \otimes \re^{d_\aa}$, written as
    \begin{align*}
       \cS=\sum_{i,j}  D^{ij}\otimes \hat{e}_i\hat{e}_j^\T\,  
    \end{align*}
 where $D^{ij}$ are non-negative matrices on $\re^{d_\s}$ and $\{\hat e_i\}$ is an orthonormal basis that spans $\re^{d_\aa}$.    Assume that 
  \begin{equation*}
       D^{ij}=\left(\begin{array}{@{}c|c@{}}
     \begin{matrix}
       A^{ij}
     \end{matrix}
     & B^{ij} \\
    \hline\rule{0in}{.14in}
     \zero &
    \begin{matrix}
     C^{ij}
    \end{matrix}
    \end{array}\right)\, \forall i,j \, ,
    \end{equation*}
where the dimensions of the blocks are the same for all $i,j$. The following hold:
\begin{enumerate}[(i)]
     \item $B^{ij} = \zero$ for all $i,j$. 
      \item $ \cS  = \oplus_\beta \cS_\beta$, with $\cS_\beta$ bistochastic matrices. 
\end{enumerate}
 
\nobreak\hfill $\square$ \end{lemma}
\begin{proof}
Let us note that 
\begin{align*}
 \cS  =  \left(\begin{array}{@{}c|c@{}}
     \begin{matrix}
       \sum_{i,j} A^{ij}\otimes \hat{e}_i\hat{e}_j^\T
     \end{matrix}
     & \sum_{i,j} B^{ij}\otimes \hat{e}_i\hat{e}_j^\T \\
     \hline\rule{0in}{.16in}
     \zero &
    \begin{matrix}
     \sum_{i,j} C^{ij}\otimes \hat{e}_i\hat{e}_j^\T
    \end{matrix}
    \end{array}\right)  =: \left(\begin{array}{@{}c|c@{}}
     \begin{matrix}
       \mathbb{A}
     \end{matrix}
     & \mathbb{B} \\
    \hline 
     \zero &
    \begin{matrix}
     \mathbb{C}
    \end{matrix}
    \end{array}\right)\,
    \, . 
\end{align*}

 By \lemref{lem:not-bistochasticity}, since $\cS$ is bistochastic then $\mathbb{A}$ and $\mathbb{C}$ are bistochastic, while  $\mathbb{B} =\zero \implies B^{ij} = \zero$ must hold for all $i,j$. It follows that $\cS  = \mathbb{A} \oplus \mathbb{C} \equiv \oplus_\beta \cS_\beta$.
    
\end{proof}

\

\begin{lemma}
\label{lem:unique-fixed-point}
    Let $\Phi$ be a channel acting in $\h$, and let $\varphi$ be an eigen-basis of a fixed state $\rho$ of $\Phi$. Assume that the $\varphi$-classical action  of $\Phi$ is irreducible. Then $\rho$ is strictly positive.
\unskip\nobreak\hfill $\square$ \end{lemma}

\begin{proof}
    If $\rho $ is diagonalisable with respect to $\varphi$, and if $\Phi(\rho) = \rho$, then $\bm{q}^\rho \equiv  [\bm{q}_m^\rho \coloneq \<\varphi_m| \rho |\varphi_m\> ] $ is a fixed point of $\cT$, i.e., $\cT \bm{q}^\rho = \bm{q}^\rho$, where $\cT$ is the $\varphi$-classical action of $\Phi$. Since $\rho$ is a state, then $\bm{q}^\rho$ is non-vanishing, i.e., $\bm{q}^\rho_m \ne 0$ for some $m$. If  $\cT$ is irreducible, then by the Perron-Frobenius theorem it has a unique non-vanishing fixed point $\bm{p}$, which is strictly positive, i.e., $\bm{p}_m >0$ for all $m$ \cite{D-Sere-Matrices}.  It follows that  $\bm{q}^\rho = \bm{p}$ must hold, and so $\rho$ must be strictly positive. 
\end{proof}

We are now ready to prove our claim, i.e., that any channel $\Phi \in \opset_\III(\hs)$ necessarily has a strictly positive fixed state:

\

\begin{prop}
\label{prop:Kraus-block-form}
    Let $\Phi : \lo(\hs) \to \lo(\hs)$ be a channel acting in $\hs$. Assume that $\Phi$  can be implemented as 
    \begin{equation}
    \label{eq:full-rank-bistochastic-app}
    \Phi(\bigcdot) = \tr\sub{\aa} [\ee ( \bigcdot \otimes \xi )]\, ,
    \end{equation}    
    where $\xi$ is a strictly positive state on $\ha$ and $\ee$ is a bistochastic channel acting in $\hs \otimes \ha$. There exists a ``maximal'' set of orthocomplete projections $\{P_\beta\}$ on $\hs$  such that the following hold:
\begin{enumerate}[(i)]
    \item For each $\beta$,  the  operation   $ \Phi_\beta(\bigcdot) \coloneq  P_\beta\Phi(P_\beta \bigcdot P_\beta) P_\beta$ satisfies $\Phi_\beta(\bigcdot) = \Phi(P_\beta \bigcdot P_\beta)$, and the restriction of $\Phi_\beta$ from $\lo(\hs) \to \lo(\hs)$ to $\lo(P_\beta \hs) \to \lo(P_\beta \hs)$, also denoted $\Phi_\beta$, is a channel. 
    \item For any state $\rho$ on $\hs$ such that $[\rho,P_\beta] = \zero$ for all $\beta$, it holds that 
    \begin{align*}
        \Phi(\rho) = \sum_\beta \Phi_\beta( \rho ) \, .
    \end{align*}
    \item For every $\beta$ and any orthonormal basis $\varphi^\beta$ that spans $P_\beta \hs$, the $\varphi^\beta$-classical action $\mathbb{T}_{\varphi^\beta}$ of  $\Phi_\beta$ is irreducible.

    \item For any $\beta$, let $\sigma_\beta$ be a state on $P_\beta \hs$ such that $\Phi_\beta(\sigma_\beta) = \sigma_\beta$. Then $\sigma_\beta$ has full rank in $P_\beta \hs$.

    \item Let $\{p_\beta\}$ be a probability distribution, and $\sigma_\beta$ fixed states of $\Phi_\beta$. Then $\rho_0 = \sum_\beta p_\beta \sigma_\beta$ is a fixed state of $\Phi$. 
    
    \item $\ff(\Phi)$ contains a strictly positive state.

\end{enumerate}
\unskip\nobreak\hfill $\square$ \end{prop}
\begin{proof}

Let $\{K_a\}_a$ be a Kraus representation of a channel $\Phi$ acting in $\hs$, i.e., $\Phi(\bigcdot) = \sum_a K_a \bigcdot K_a^*$. There exists a set of orthocomplete projections $\{P_\beta\}$ such that  $K_a=\sum_\beta P_\beta K_a P_\beta$ holds for all $a$, which is equivalent to $[K_a, P_\beta] = \zero$ for all $a, \beta$. To see that such a set of projections always exists, note that by choosing $\{\beta\}$ as a singleton, so that $P_\beta = \onesys$, the above properties trivially hold.  It is trivial to see that $\{K_a^\beta\}_a$, where $K_a^\beta \coloneq P_\beta K_a P_\beta$, is a Kraus representation for the  operation  $\Phi_\beta (\bigcdot) \coloneq P_\beta \Phi( P_\beta \bigcdot P_\beta) P_\beta$. 

Now we shall prove item (i). That $\Phi_\beta(\bigcdot) =  \Phi(P_\beta \bigcdot P_\beta)$ follows  immediately from the fact that $P_\beta$ are projections and that $[K_a, P_\beta] = [K_a^*, P_\beta] = \zero$ for all $a$. To see that $\Phi_\beta$ is a channel when restricted to  $\lo(P_\beta \hs ) \to \lo(P_\beta \hs )$,  it is sufficient to note that the unit in $P_\beta \hs$ is the projection $P_\beta$, and that 
\begin{align*}
    \Phi_\beta^*(P_\beta) = \sum_a {K^\beta_a}^* P_\beta {K^\beta_a}  = \sum_a P_\beta K_a^* P_\beta K_a P_\beta  = \sum_a K_a^* K_a P_\beta  = P_\beta \, .
\end{align*}
Here, we have used the fact that $[K_a, P_\beta] = [K_a^*, P_\beta] = \zero$ for all $a$,  that $P_\beta$ is a projection, and that $\sum_a K_a^* K_a = \Phi^*(\onesys) = \onesys$. 

Now we  prove item (ii). Consider a state $\rho$ such that  $[\rho, P_\beta] = \zero \, \forall \beta$, which implies that $\rho = \sum_\beta P_\beta \rho P_\beta$. Since $\Phi_\beta(\bigcdot) =\Phi(P_\beta \bigcdot P_\beta)$, i.e., if the input of $\Phi$ is in $P_\beta \hs$ then the output is guaranteed to also be in $P_\beta \hs$,  it holds that
\begin{align*}
    \Phi(\rho) = \sum_\beta \Phi(P_\beta \rho P_\beta) = \sum_\beta \Phi_\beta(\rho) \, . 
\end{align*}

Now we prove item (iii), i.e., show that if $\Phi$ can be implemented by a bistochastic interaction with a strictly positive auxiliary system as in \eq{eq:full-rank-bistochastic-app}, then  a ``maximal'' set of orthocomplete projections $\{P_\beta\}$ exists such that the classical action of the channel $\Phi_\beta$ acting in $P_\beta \hs \subseteq \hs$ is irreducible for any ONB that spans the subspace $P_\beta \hs$. 

Denote by $\varphi^\beta\coloneq\{\ket{\varphi_{i_\beta}^\beta}\}_{i_\beta=1}^{d_\beta}$, where $d_\beta = \rank{P_\beta}$,  any  orthonormal basis that spans $P_\beta \hs \equiv \supp(P_\beta)$, i.e.,  $\< \varphi_{j_{\beta'}}^{\beta'}| \varphi_{i_\beta}^{\beta}\> = \delta_{\beta, \beta'}\delta_{i_\beta,j_\beta}$. Then any orthonormal basis  $\varphi = \{|\varphi_m\>\}_m$ that spans $\hs$ can be constructed as  $\varphi= \cup_\beta \varphi^\beta  = \{\ket{\varphi^\beta_{i_\beta}}\}_{\beta, i_\beta}$.  

The Kraus operators $K_a$, in the $\varphi$-matrix representation, read
\begin{equation}
    \label{eq:block-kraus}
     K_a=\bigoplus_\beta K_a^\beta\, ,
\end{equation}
where $K_a^\beta$ are matrices in the $\varphi^\beta$ representation.   On the one hand, by Eqs.~\eqref{eq:T-kraus} and \eqref{eq:block-kraus}, the $\varphi$-classical action of $\Phi$, i.e.,  $\cT$, is given by
    \begin{equation}
    \label{eq:T-direct-sum}
        \cT=\sum_a K_a\odot\overline{K}_a=\bigoplus_\beta \cT_{\varphi^\beta}\, ,
    \end{equation}
where $\cT_{\varphi^\beta}=\sum_a K_a^\beta\odot\overline{K}_a^\beta$ is the $\varphi^\beta$-classical action of $\Phi_\beta$. On the other hand, by definition \eqref{eq:classical-action} and using the eigen-decomposition $\xi=\sum q_j \proj{\psi_j}$ in \eq{eq:full-rank-bistochastic-app}, the matrix elements of $\cT$ read
     \begin{align}
    \label{eq:T-and-S-classical-actions}
        \cT_{m,n}&=\bra{\varphi_m}\Phi\left(\proj{\varphi_n}\right)\ket{\varphi_m}\nonumber\\
        &=\bra{\varphi_m}\tr\sub{\aa}\left[\ee\left(\proj{\varphi_n}\otimes\sum q_j\proj{\psi_j}\right)\right]\ket{\varphi_m}\nonumber\\
        &=\sum_{i,j}q_j\bra{\varphi_m}\bra{\psi_i}\ee\left(\proj{\varphi_n}\otimes\proj{\psi_j}\right)\ket{\varphi_m}\ket{\psi_i}\nonumber\\
        &=\sum_{i,j}q_j\cS_{m, n; i, j}\, ,
    \end{align}
where $\cS$ is the $(\varphi\psi)$-classical action of $\ee$, and $q_j>0$ for all $j$ by the assumption of strictly positivity of $\xi$. Since $\ee$ is a bistochastic channel acting in $\hs \otimes \ha$, then $\cS$ is a bistochastic matrix, which  can be written as 
    \begin{equation}
    \label{eq:S-block-structure}
        \cS=\sum_{i,j}  D^{ij}\otimes \hat{e}_i\hat{e}_j^\T\, ,
    \end{equation}
where $\{\hat{e}_i\}$ is an orthonormal basis spanning $\mathbb{R}^{d_\aa}$, with $d_\aa = \dim(\ha)$,  and the matrices $D^{ij}$ are entry-wise non-negative and of dimension $d_\s = \dim(\hs)$. Noting that here, the matrix elements of $\cS$ read
\begin{align*}
\cS_{m, n ; i, j } = D^{ij}_{m,n}\, ,   
\end{align*}
Eqs. \eqref{eq:T-and-S-classical-actions} and \eqref{eq:S-block-structure}  imply that 
    \begin{equation*}
  %  \label{eq:T-to-D}
        \cT_{m,n}=\sum_{ij}q_jD^{ij}_{m,n} \, .
    \end{equation*}
Given that $\cT$ has the block form of Eq.~\eqref{eq:T-direct-sum}, then due to the entry-wise non-negativity of $D^{ij}$ and positivity of $q_j$ for all $j$, it follows that for all $i,j$,  the matrices $D^{ij}$  must also admit this block-diagonal structure. That is to say, 
    \begin{equation*}
        D^{ij}=\bigoplus_\beta D^{ij}_\beta\, ,
    \end{equation*}
where $D^{ij}_\beta$ are entry-wise non-negative matrices of dimension $d_\beta$. This has two consequences. First, due to Eq.~\eqref{eq:S-block-structure} and \lemref{lemma:bipartite-bistochastic-reducibility}, it holds that
\begin{align*}
 \cS  =\bigoplus_\beta \cS_\beta\, ,    
\end{align*}
with each block 
    \begin{equation}
    \label{eq:blocks-of-S}
        \cS_\beta=\sum_{ij}  D^{ij}_\beta\otimes \hat{e}_i\hat{e}_j^\T\, 
    \end{equation}
of dimension $ d_\beta d_\aa$ being bistochastic itself. Second, the $\varphi^\beta$-classical action of $\Phi_\beta$ is given by
    \begin{equation}
    \label{eq:beta-block-of-T}
        \cT_{\varphi^\beta}=\sum_{ij} q_j D^{ij}_\beta\, .
    \end{equation}
Now, we proceed with the proof of irreducibility of $\cT_{\varphi^\beta}$ by contradiction. Specifically, we will show that if $\cT_{\varphi^\beta}$ is reducible, then  $\{P_\beta\}$ is not maximal, in the sense that $P_\beta$ can be decomposed into  smaller orthogonal projections $P_{\beta, \alpha}, \sum_\alpha P_{\beta, \alpha} = P_{\beta}$. Assume that there exist a subspace $P_\beta\hs$ and an orthonormal basis $\varphi^\beta$ spanning it such that  $\cT_{\varphi^\beta}$ is reducible and can be brought into the form of \eq{eq:block_form_matrix}. This means that for all $i,j$, the matrix $D_{\beta}^{ij}$ has this block form, i.e., up to some permutation $\Pi_\beta$ of the basis $\varphi^\beta$ we may write 
    \begin{equation*}
%   \label{eq:D-block_form}
       D^{ij}_\beta=\left(\begin{array}{@{}c|c@{}}
     \begin{matrix}
       A^{ij}
     \end{matrix}
     & B^{ij} \\
    \hline\rule{0in}{.14in}
     \zero &
    \begin{matrix}
     C^{ij}
    \end{matrix}
    \end{array}\right)\, \forall i,j \, ,
    \end{equation*}
where the dimensions of the blocks are the same for all $i,j$.    By \eq{eq:blocks-of-S} and \lemref{lemma:bipartite-bistochastic-reducibility}, it must hold that $B^{ij} = \zero$ for all $i,j$. As such,  we have that $D_\beta^{ij} = \oplus_\alpha D_{\beta,\alpha}^{ij}$ for all $i,j$, where $D_{\beta,1}^{ij} = A^{ij}$ and $D_{\beta,2}^{ij} = C^{ij}$,  and so by \eq{eq:beta-block-of-T} it holds that $\cT_{\varphi^\beta} = \oplus_\alpha \cT_{\varphi^{\beta,\alpha}}$. By \eq{eq:T-kraus}, this implies that $K_{a}^\beta = \oplus_\alpha K_{a}^{\beta,\alpha}$, and so  there is a smaller, or a more refined, orthocomplete set of projections $P_{\beta, \alpha}$ for which items (i) and (ii) hold.   

Now we  prove item (iv). By item (i)  the operation $\Phi_\beta$ is a channel  acting in $P_\beta \hs$. Due to the Schauder–Tychonoff fixed point theorem, all channels acting in a finite-dimensional Hilbert space have at least one fixed state in that space. As such, there exists a state $\sigma_\beta \in \lo(P_\beta \hs)$ such that $\Phi_\beta (\sigma_\beta) = \sigma_\beta$. Let $\varphi^\beta$ be an eigenbasis of $\sigma_\beta$, and let $\cT_{\varphi^\beta}$ be the $\varphi^\beta$-classical action of $\Phi_\beta$.   By item (iii), the classical action of $\Phi_\beta$ is irreducible for any ONB that spans $P_\beta \hs$. It follows that $\cT_{\varphi^\beta}$ is irreducible.  By \lemref{lem:unique-fixed-point}, it follows that $\sigma_\beta$ must be strictly positive, i.e., it has full rank in $P_\beta \hs$. 

Now we  prove item (v). Consider the convex combination $\rho_0 = \sum_\beta p_\beta \sigma_\beta$.   Since $\sigma_\beta = P_\beta \sigma_\beta$, then $[\rho_0,P_\beta]= \zero \, \forall \beta \iff \rho_0 = \sum_\beta P_\beta \rho_0 P_\beta$ holds. By item (ii)  it holds that
\begin{align*}
    \Phi(\rho_0) &= \sum_\beta \Phi_\beta(\rho_0) = \sum_\beta p_\beta \Phi_\beta(\sigma_\beta) = \sum_\beta p_\beta \sigma_\beta = \rho_0 \,.
\end{align*}
The second equality follows from the fact that $\Phi_\beta(\rho_0) = \Phi_\beta(P_\beta \rho_0 P_\beta) = p_\beta \Phi_\beta(\sigma_\beta)$, and  the penultimate equality follows from the fact that $\Phi_\beta(\sigma_\beta) = \sigma_\beta$. 

Finally, we shall prove item (vi). If in the above we choose $p_\beta >0$,  $\rho_0$ has full rank in $\hs$. It follows that $\Phi$ has a strictly positive  fixed state $\rho_0$.

\end{proof}

\section{Fixed-points of  measurement channels consistent with the strong third law}\label{app:fixed-point-measurement}

In this section, we shall provide the full proof for item (v) of \thmref{thm:non-ddisturbing-thermo-nogo}. To this end, we need to explore in more depth the properties of the fixed points of the $\E$-channel $\ii_\xx$, when $\ii$ is consistent with the strong third law, i.e., when $\ii \in \inset_\II(\hs)$. 

\

\begin{lemma}\label{lemma:fixed-point-rank-non-decreasing}
 Let  $(\ha, \xi, \ee, \Z)$ be a  process  for an $\E$-compatible instrument $\ii$ acting in $\hs$.  Consider a state $\rho \in \ff(\ii_\xx)$, with  support projection $P$. If  $\ee$ is rank non-decreasing and $\xi$ is strictly positive,   the following hold:
 \begin{enumerate}[(i)]

      \item   $\supp(\ee(\rho \otimes \xi)) = P\hs \otimes \ha$, i.e., $\ee(\rho \otimes \xi)$ has full rank in $P\hs \otimes \ha$.
      \item If $\rank{E_x} =1$ for some $x$,  then   $\rho > \zero$. 
 \end{enumerate}
\unskip\nobreak\hfill $\square$ \end{lemma}
\begin{proof}
  \begin{enumerate}[(i)]
      \item 
    Since $\rank{\xi} = \dim(\ha)$ and $\ee$ is rank non-decreasing, we may write
    \begin{align}\label{eq:strong-third-fixedstate-inequality}
      \rank{\rho}\dim(\ha) & \leqslant \rank{\ee(\rho \otimes \xi)} \leqslant \rank{\ii_\xx(\rho)}\rank{\Lambda(\rho)} \, ,    
     \end{align}
      where $\Lambda(\bigcdot)  := \trs[\ee(\bigcdot \otimes \xi)]$ is the conjugate channel to $\ii_\xx$, and  the final inequality follows from  \lemref{lemma:rank-partial-trace}. 
      Since  $\ii_\xx(\rho) = \rho \implies \rank{\ii_\xx(\rho)} = \rank{\rho} = \dim( P \hs)$, it follows that  $\rank{\Lambda(\rho)} = \dim(\ha)$. Therefore, it also holds that $\rank{\ee(\rho \otimes \xi)} =\dim(\P\hs) \dim(\ha) = \dim(\P\hs\otimes \ha)$. 

Now, since $P \rho = \rho$, then  $\rho = \ii_\xx(\rho)$ implies that 
   \begin{align*}
       \tra[(P \otimes \oneapp ) \, \ee(\rho \otimes \xi)] = P \, \tra[ \ee(\rho \otimes \xi)] = P \ii_\xx(\rho) = P \rho = \rho \, .
   \end{align*} 
But this implies that $\tr[(P \otimes \oneapp ) \, \ee(\rho \otimes \xi)] = 1$, and so by \lemref{lemma:effect-support-probability-1} it must hold that $(P \otimes \oneapp) \ee(\rho \otimes \xi) = \ee(\rho \otimes \xi)$. That is, $\supp(\ee(\rho \otimes \xi)) \subset P\hs \otimes \ha$.  But as shown above, $\rank{\ee(\rho \otimes \xi)} = \dim(\P\hs\otimes \ha)$, and so this implies that   $\supp(\ee(\rho \otimes \xi)) = P\hs \otimes \ha$, i.e.,  $\ee(\rho \otimes \xi)$ has full rank in $P\hs \otimes \ha$.

 \item 
 
   Since $E_x \ne \zero \implies Z_x \ne \zero$, by item (i) and the probability reproducibility condition  it holds that 
 \begin{align*}
   \tr[E_x \rho] =  \tr[\onesys \otimes Z_x\ \ee(\rho \otimes \xi)] \equiv \tr[Z_x \Lambda(\rho)] >0 \, .
\end{align*}
 
 Let $\rho$ be a fixed state of the channel $\ii_\xx$.   By item (ii), it holds that $ \tr[ E_x  \rho] >0 $ for all $x$. Let $E_x = \lambda_x P_x$ be a rank-1 effect. Then $\tr[E_x \rho]>0$ implies that $P_x \rho P_x = \tr[P_x \rho] P_x$ has full rank in the 1-dimensional subspace $P_x \hs$, and so by item (iii) of \lemref{lemma:measurement-weak-third-law} it follows that $\sigma\coloneq \ii_x(\rho)/\tr[E_x \rho]$ has full rank in $\hs$. As such, we have that $\rho = \ii_\xx(\rho) = \tr[E_x \rho] \sigma + \sum_{y \ne x} \ii_y(\rho)$, and since a mixture of a full-rank state with any other state must be full-rank,  it must hold that $\rho$ has full rank in $\hs$. That is, if $E_x$ has rank 1 for some $x$, then $\rho \in \ff(\ii_\xx) \implies \rank{\rho} = \dim(\hs)$. 
 \end{enumerate}   
\end{proof}  

Before proceeding further, let us recall some results shown previously in Appendix M of Ref. \cite{Mohammady2021a}.  We define the ``average'' of the $\E$-channel $\ii_\xx$ and its dual as
\begin{align}\label{eq:av-channel}
& \ii_{\av}(\bigcdot) \coloneq\lim_{N\to \infty} \frac{1}{N}\sum_{n=1}^{N}(\ii_\xx)^n(\bigcdot) \, , 
    & \ii^*_{\av}(\bigcdot)\coloneq\lim_{N\to \infty} \frac{1}{N}\sum_{n=1}^{N}(\ii^*_\xx)^n(\bigcdot) \, .
\end{align}
$\ii^*_{\av}$ is a (unital) CP projection on $\ff(\ii_\xx^*) = \ff(\ii_\av^*)$, i.e., it holds that $\ii_\av^* = \ii_\av^* \circ \ii_\xx^* = \ii_\xx^* \circ \ii_\av^* = \ii_\av^* \circ \ii_\av^*$. Similarly,  $\ii_\av$ is  a CP projection on $\ff(\ii_\xx) = \ff(\ii_\av)$.  Let $\P$ denote the minimal projection on the fixed-point set $\ff(\ii_\xx)$, that is, for all projections $Q$ such that $\rho = Q \rho $ for all $\rho \in \ff(\ii_\xx)$, it holds that $\P \leqslant Q$. $\P$ equals the support projection for the state 
\begin{align}\label{eq:complete-mixture-av-channel}
    \rho_0 \coloneq \ii_{\av} \left( \frac{\onesys}{\dim(\hs)}\right) \, .
\end{align}
Note that  $\P=\onesys$ if and only if $\ff(\ii_\xx)$ contains a strictly positive state. 
We may use $\P$ to define the CP maps 
\begin{align}\label{eq:av-channel-P}
& \ii_\P^*(\bigcdot) \coloneq \P \ii_\xx^*(\bigcdot)\P, & \ii_{\av,\P}^*(\bigcdot) \coloneq \P \ii_\av^*(\bigcdot)\P.
\end{align} 
These maps are unital (with unit $\P$) when the domain and image are restricted from $\lo(\hs)$ to $\lo(\P \hs)$. Indeed, we  observe that 
\begin{align}\label{eq:av-channel-P-identity}
    \ii_\av^*(\bigcdot) = \ii_\av^*(\P \bigcdot \P) \, \, , \, \,  \ii_{\av,\P}^*(\bigcdot) =  \ii_{\av,\P}^*(\P \bigcdot \P) \, \, , \, \,  \ii_{\P}^*(\bigcdot) =  \ii_{\P}^*(\P \bigcdot \P) \, . 
\end{align}

The fixed points of these  CP maps  are  defined as 
\begin{align*}
& \ff(\ii_P^*) \coloneq \{ A \in \lo(\P\hs) :  \ii_\P^*(A) = A\} & \ff(\ii_{\av,P}^*)  \coloneq  \{ A \in \lo(\P\hs) :  \ii_{\av,\P}^*(A) = A\}  \, ,
\end{align*}
and we observe that 
\begin{align}\label{eq:algebra-fixed-points}
\P\ff(\ii_\xx^*) \P & : = \{ \P A \P : A \in \ff(\ii_\xx^*)\}  \equiv  \ff(\ii_P^*)  \equiv \ff(\ii_{\av,P}^*)  \, . 
\end{align}
That is, for any $A \in \ff(\ii_\xx^*)$, it holds that $\P A \P \in \ff(\ii_P^*)  \equiv \ff(\ii_{\av,P}^*)$. Similarly, for any $A \in \ff(\ii_P^*)  \equiv \ff(\ii_{\av,P}^*)$, there exists $B \in \ff(\ii_\xx^*)$ such that $\P B \P = A$. 

Since there exists a state $\rho_0$  that has full rank in  $\P\hs$ which is non-disturbed by $\ii_\xx$, it follows that $\ff(\ii_\P^*) \equiv  \ff(\ii_{\av,\P}^*)$ is a  von Neumann algebra \cite{Bratteli1998, Arias2002}, i.e., $\ff(\ii_\P^*)$ satisfies multiplicative closure. But since $\P \hs$ is finite-dimensional, then $\ff(\ii_\P^*)$ is a finite von Neumann algebra $\mathscr{A}$, which may have an Abelian non-trivial center $\mathscr{Z} \coloneq \mathscr{A} \cap \mathscr{A}'$ generated by the set of ortho-complete projections $\{P_\alpha\}$ which satisfy $\sum_\alpha P_\alpha = \P$. That is, every self-adjoint $B \in \mathscr{Z}$ can be written as $B = \sum_\alpha  \lambda_\alpha P_\alpha$.  We may therefore decompose $\mathscr{A}$ into a finite direct sum $\mathscr{A} = \oplus_\alpha \mathscr{A}_\alpha$, where each $\mathscr{A}_{\alpha} = P_\alpha \mathscr{A}$ is a 
type-I factor (a finite dimensional von Neumann algebra with a trivial center) on 
$P_{\alpha}\hs = \kk_\alpha \otimes \rr_\alpha$, 
written as $\mathscr{A}_{\alpha} = \lo(\kk_\alpha)\otimes \one\sub{\rr_\alpha}$. Note that here, $P_\alpha = \one\sub{\kk_\alpha} \otimes \one\sub{\rr_\alpha}$. It follows  that  we may write 

\begin{align}\label{eq:factor-fixed-point-set}
\ff(\ii_\xx) &= \bigoplus_\alpha \lo(\kk_\alpha) \otimes \omega_\alpha  \, , \nonumber \\
\ff(\ii_P^*) \equiv \ff(\ii_{\av,\P}^*)&= \bigoplus_{\alpha} \lo(\kk_{\alpha}) \otimes \one\sub{\rr_\alpha}  \, ,  
\end{align}
and 
\begin{align}\label{eq:factor-av-channel}
\ii_\av(\bigcdot) &= \sum_{\alpha}
\tr\sub{\rr_\alpha}[P_{\alpha} \bigcdot P_{\alpha}]\otimes \omega_{\alpha}  \, ,  \nonumber \\
\ii_{\av,\P}^* (\bigcdot) &= \sum_{\alpha}\Gamma_{\omega_{\alpha}}
(P_{\alpha} \bigcdot P_{\alpha})\otimes \one_{\mathcal{R}_{\alpha}}  \, , 
\end{align}
where:  $\omega_\alpha$ are states on $\rr_\alpha$;   $\Gamma_{\omega_\alpha} : \lo(\kk_{\alpha}\otimes \rr_{\alpha}) \to \lo(\kk_\alpha)$ are restriction maps; and $\tr\sub{\rr_\alpha}: \lo(\kk_\alpha \otimes \rr_\alpha) \to \lo(\kk_\alpha)$ are partial traces \cite{Lindblad1999a}. Note that  $\omega_\alpha$ are  states with full rank in $\rr_\alpha$. This is because the state $\rho_0$ defined in \eq{eq:complete-mixture-av-channel}   has full rank in $\P \hs$. But since   $\rho_0 := \ii_\av(\onesys/\dim(\hs)) \propto \oplus_\alpha \one\sub{\kk_\alpha} \otimes \omega_\alpha$, then $\rho_0$ has full rank in $\P \hs$ if and only if $\omega_\alpha$ have full rank in $\rr_\alpha$ for all $\alpha$.

We now provide a useful result indicating the form that the effects of $\E$ must take in light of the fixed-point structure of the $\E$-channel $\ii_\xx$.  This is a generalisation of Lemma E.1 in Ref. \cite{Mohammady2022a}, which holds if $\ii$ is constrained by the weak third law, i.e., $\ii \in \inset_\I(\hs)$,  \emph{and} if $\ff(\ii_\xx)$ contains a strictly positive state. 

\

\begin{lemma}\label{lemma:effect-factor-decomposition}
 Let $\E \coloneq \{E_x : x \in \xx\}$ be a non-trivial observable on $\hs$, and let  $(\ha, \xi, \ee, \Z)$ be a measurement process  for an $\E$-compatible instrument $\ii$ acting in $\hs$. Let $\P$ be the minimal support projection on $\ff(\ii_\xx)$, and define the restriction of observable $\E$ in $\P \hs$ as
 \begin{align*}
     \P \E \P := \{\P E_x \P : x \in \xx\} \, .
 \end{align*}
The following hold: 
\begin{enumerate}[(i)]
    \item $\P \ff(\ii_\xx^*) \P \subset (\P \E \P)'$.

    \item If $\ee$ is rank non-decreasing and $\xi$ is strictly positive, then  
 \begin{align}\label{eq:effect-direct-sum-decomp}
    \P E_x \P = \bigoplus_\alpha \one\sub{\kk_\alpha} \otimes E_{x, \alpha}  \, ,
\end{align}
where  for all $x$ and $\alpha$,   $\zero \leqslant  E_{x, \alpha} \leqslant \one\sub{\rr_\alpha}$,  $E_{x,\alpha} \ne \zero$,  and $E_{x,\alpha} \ne \one\sub{\rr_\alpha}$. 
\end{enumerate}
\unskip\nobreak\hfill $\square$ \end{lemma}

\begin{proof}
Using the CP unital map $\Gamma^\ee_\xi$ defined in \eq{eq:gamma-E}, let us define the CP subunital map $\Gamma^\ee_{\xi,\P} : \lo(\hs \otimes \ha) \to \lo(\P \hs)$ as $\Gamma^\ee_{\xi, \P}(\bigcdot) := \P \Gamma^\ee_\xi(\bigcdot) \P$.  We may write  $\P \ii_x^*(\bigcdot) \P = \Gamma_{\xi, \P}^\ee (\bigcdot \otimes Z_x)$, and so    $\P E_x \P = \P \ii_x^*(\onesys) \P =  \Gamma_{\xi,\P}^\ee (\onesys \otimes Z_x) $. Similarly, we may write $\ii_\P^*(\bigcdot) \coloneq \P \ii_\xx^*(\bigcdot) \P = \Gamma_{\xi,\P} ^\ee (\bigcdot \otimes \oneapp)$.  Since the fixed-point set $\ff(\ii_\P^*) \equiv \ff(\ii_{\av,\P}^*) \subset \lo(\P \hs)$  is a von Neumann algebra, for any $A \in \ff(\ii_\P^*)$ it holds that $A^*A, A A^* \in \ff(\ii_\P^*)$.     By the multiplicability theorem \cite{Choi1974}, this implies that   $A \Gamma_{\xi, \P}^\ee(B) = \Gamma_{\xi, \P}^\ee ( (A\otimes \oneapp) B) $ and $\Gamma_{\xi, \P}^\ee (B)  A = \Gamma_{\xi, \P}^\ee ( B (A\otimes \oneapp)) $ for all $A \in \ff(\ii_\P^*)$ and $B \in \lo(\hs \otimes \ha)$. By choosing $B = \onesys \otimes Z_x$,  we may therefore write
\begin{align*}\label{eq:multiplicative-non-disturbance-Q}
\P \ii_x^*(A) \P =  \Gamma_{\xi, \P}^\ee (A \otimes Z_x)  =  A \P E_x \P = \P E_x \P A
\end{align*}
for all $A \in \ff(\ii_\P^*)$. That is, 
\begin{align*}
\P \ff(\ii_\xx^*) \P = \ff(\ii_\P^*) \subset (\P\E \P)' \coloneq \{A \in \lo(\P \hs) : [\P E_x \P, A] = \zero \, \, \forall x \in \xx \}  \, ,    
\end{align*} 
i.e., the fixed points of $\ii_\P^*$  are contained in the commutant of $\P \E \P$ in $\P \hs$. Equivalently, for any $A \in \ff(\ii_\xx^*)$, the restriction $\P A \P$ 
 is contained in the commutant of $\P \E \P$ in $\P \hs$. This concludes the proof for item (i). 
 
 Now we shall proceed with proving item (ii). Note that the condition $\ff(\ii_\P^*)  \subset (\P \E \P)'$ implies that $\P\E\P \subset \ff(\ii_\P^*) '$. By    \eq{eq:factor-fixed-point-set}, we have that 
\begin{align*}
 \ff(\ii_\P^*) ' =    \bigoplus_{\alpha} \one\sub{\kk_{\alpha}} \otimes \lo(\rr_\alpha)  \, .
\end{align*}
That the effects of $\P \E \P$ are decomposed as in \eq{eq:effect-direct-sum-decomp} directly follows. Moreover, $\zero \leqslant  E_{x, \alpha} \leqslant \one\sub{\rr_\alpha}$ follows trivially from the fact that $E_x$, and hence $\P E_x \P$, are effects. So now we shall show that $E_{x,\alpha} \ne \zero$ and $E_{x,\alpha} \ne \one\sub{\rr_\alpha}$.

Note that for any $A \in \ff(\ii_\P^*)$, outcome $x$,  and state $\rho_0$ that has full rank in $\P \hs$, it holds that $ \tr[\rho_0 A^*A \P E_x \P] =  \tr[\rho_0 (A \sqrt{\P E_x \P})^*(A \sqrt{\P E_x \P})] \geqslant 0$, which vanishes if and only if $A^*A \P E_x \P = \zero$. But now we may write the following:
\begin{align*}
 \tr[\rho_0 A^*A \P E_x \P]  &=   \tr[\rho_0 \P \Gamma_\xi^\ee (A^*A \otimes Z_x)\P] \\
 & = \tr[\ee(\rho_0 \otimes \xi) A^*A \otimes Z_x] \\
 & \geqslant 0  \, .
\end{align*}
Since $\xi$ has full rank in $\ha$ and $\ee$ is rank non-decreasing, then by \lemref{lemma:fixed-point-rank-non-decreasing}  it follows that for a fixed state $\rho_0 = \ii_\xx(\rho_0)$ that has full rank in $\P \hs$, it holds that $\ee(\rho_0 \otimes \xi)$ has full rank in $\P\hs \otimes \ha$. Since $ Z_x \ne   \zero$, then the equality condition of the above equation is satisfied for such a state if and only if $A=\zero$. Therefore,     $A^*A \P E_x \P =  \zero \iff A = \zero$. Since $A \P E_x \P =\zero \implies A^*A \P E_x \P = \zero$, it follows that $A \P E_x \P =\zero \iff A =\zero$.

Now assume that $E_{x,\alpha} = \zero$ for some $\alpha$. It will hold that an operator  $A = A_\alpha \otimes \one\sub{\rr_\alpha} \in \ff(\ii_{\P}^*)$ exists,   with $ A_\alpha \ne \zero$, such that   $A \P E_x\P = \zero$. But this contradicts what we showed above. Therefore, all $E_{x,\alpha}$ must be non-vanishing. Finally, since $\E$ is non-trivial, then there exists at least two distinct outcomes, and so by normalisation it holds that  $E_{x,\alpha} \ne    \one\sub{\rr_\alpha}$. 

\end{proof}

We are now ready to prove items  (v) and (vi) of \thmref{thm:non-ddisturbing-thermo-nogo} in the main text, which we reiterate here for convenience:

\

\begin{theorem}\label{thm:non-ddisturbing-thermo-nogo-app}
Consider an $\E$-compatible instrument $\ii:= \{\ii_x : x\in \xx\}$ acting in $\hs$, and assume that  $\E$ is a non-trivial observable. Assume that $\ii$ belongs to $\inset_C(\hs)$ for  $C \in \{\II, \III\}$ as given in \defref{defn:thermo-consistent-instrument}.  The following hold:
\begin{enumerate}[(i)]

    \item 
    
    If $\ii \in \inset_\II(\hs)$, then $\ii$ is first-kind only if $\zero < E_x < \onesys $ for all $x \in \xx$.

    \item 

    If $\ii \in \inset_\III(\hs)$, then $\ii$ is first-kind only if $\zero < E_x < \onesys $ for all $x \in \xx$ and $[E_x, E_y] = \zero$ for all $x,y \in \xx$.

\end{enumerate}

\nobreak\hfill $\square$ \end{theorem}
\begin{proof}
\begin{enumerate}[(i)]

    \item 
    
    An $\E$-compatible instrument $\ii$ is a measurement of the first kind if $\E \subset \ff(\ii_\xx^*)=\ff(\ii_\av^*)$, where $\ii_\av^*$ is a unital CP map defined in \eq{eq:av-channel}.  Now let $\P$ be the minimal support projection on the fixed states of $\ii_\xx$, and define the unital CP map $\ii_{\av, \P}^*(\bigcdot) \coloneq \P \ii_\av^*(\bigcdot) \P$. Recall from \eq{eq:algebra-fixed-points} that $\E \subset \ff(\ii_\xx^*) \equiv \ff(\ii^*_\av) \implies \P E_x \P =  \ii_{\av,\P}^*(\P E_x \P) $.    By \lemref{lemma:effect-factor-decomposition} and \eq{eq:factor-av-channel},  it follows that
    
    \begin{align*}
        \P E_x \P &= \ii_{\av,\P}^*(\P E_x \P) \\
        & =  \sum_{\alpha}\Gamma_{\omega_{\alpha}}
(P_{\alpha} \P E_x \P P_{\alpha})\otimes \one_{\mathcal{R}_{\alpha}} \\
& = \bigoplus_{\alpha}\lambda_\alpha(x) \one\sub{\kk_\alpha} \otimes \one\sub{\rr_\alpha}\equiv  \bigoplus_{\alpha}\lambda_\alpha(x) P_\alpha \, ,
    \end{align*}
   where $\lambda_{\alpha}(x)\coloneq  \tr[E_{x,\alpha} \omega_\alpha]$. Since for all $\alpha$, $\omega_\alpha$ are strictly positive states on $\rr_\alpha$, while for all $\alpha$ and  $x$, $E_{x,\alpha}$ are effects on $\rr_\alpha$ which satisfy  $E_{x,\alpha} \ne \zero$ and $E_{x, \alpha} \ne \one\sub{\rr_\alpha}$, then  $0 < \lambda_{\alpha}(x) < 1$. Now recall from \eq{eq:av-channel-P-identity} that  $\ii_\av^*(\bigcdot) = \ii_\av^*( \P \bigcdot \P)$.  It follows that if $\ii$ is a measurement of the first kind, then for all $x$ it must hold that 
\begin{align*}
    \| E_x\| = \| \ii_\av^*(E_x) \| = \| \ii_\av^*( \P E_x \P) \| \leqslant \| \P E_x \P \| <1 \, ,
\end{align*}
where the first inequality follows from the fact that $\ii_\av^*$ is CP and unital, and the final inequality follows from the fact that $\lambda_{\alpha}(x) <1$. Similarly, we may write
\begin{align*}
    \| \onesys -  E_x\| & = \| \ii_\av^*(\onesys - E_x) \| = \| \ii_\av^*(\P - \P E_x \P) \|  \leqslant \| \P - \P E_x \P \| <1 \, ,
\end{align*}
where the final inequality follows from the fact that $\lambda_{\alpha}(x) >0$. It follows that $E_x$ cannot have eigenvalue 1 or eigenvalue 0, and so $\E$ must be indefinite, i.e., $\zero < E_x < \onesys$. 

\item The requirement that $\zero < E_x < \onesys$ must hold follows immediately from (i) and the fact that $\inset_\III(\hs) \subset \inset_\II(\hs)$. The requirement that $\E$ must be  commutative, i.e., $[E_x, E_y]= \zero$, follows from  \thmref{thm:channel-III-sp-fixed-state} which states that if $\ii \in \inset_\III(\hs)$ then $\ff(\ii_\xx)$ contains a strictly positive state, and so $\ff(\ii_\xx^*)$ is a von Neumann algebra.  Indeed, note that as shown in \lemref{lemma:effect-factor-decomposition}, If $\E \subset \ff(\ii^*_\xx)$, then $\P \E \P \subset (\P \E \P)'$, where $\P$ is the minimal support projection on $\ff(\ii_\xx)$. If $\ff(\ii_\xx)$ contains a strictly positive state, then $\P = \onesys$, and it follows that $\E \subset \E'$ must hold, i.e., $\E$ must be commutative.

\end{enumerate}

\end{proof}

\begin{corollary}\label{cor:rank-1-disturbance}
  Consider an $\E$-compatible instrument $\ii \in \inset_\II(\hs)$, and assume that for some outcome $x$, the effect $E_x$ has rank 1, i.e., $E_x = \lambda \proj{\psi}$ for some unit vector $|\psi\>$ in $\hs$. It follows that $\ff(\ii_\xx^*) = \co \onesys$. That is, $\ii$ disturbs all non-trivial observables.  
\unskip\nobreak\hfill $\square$ \end{corollary}
\begin{proof}
By item (ii) of \lemref{lemma:fixed-point-rank-non-decreasing},   $\ff(\ii_\xx)$ contains a strictly positive state.  By \lemref{lemma:effect-factor-decomposition}, and inserting $\P = \onesys$, the rank of every effect of $\E$ is bounded as $\rank{ E_x } \geqslant \sum_\alpha \dim(\kk_\alpha)$.  Therefore, if any effect of $\E$ is rank-1, then it must hold that the number of indices $\alpha$ is 1, and that  $\dim(\kk_\alpha) = 1$, so that by \eq{eq:factor-fixed-point-set} we have $\ff(\ii_\xx^*) = \co \onesys$. 
\end{proof}

\bibliography{Projects-Thermal-Measurement.bib}

\end{document}